\documentclass[smallextended]{svjour3}

\smartqed  

\usepackage{amsmath}
\usepackage{amssymb}
\usepackage{graphicx}
\usepackage{float}
\usepackage{color}
\usepackage[ruled]{algorithm2e}


\usepackage[margin=1.5in]{geometry}
\newcommand{\R}{\mathbb{R}}
\newcommand{\E}{\mathbb{E}}

\newcommand{\prob}{\mathbb{P}}
\newcommand{\C}{\mathcal{C}}
\newcommand{\F}{\mathcal{F}}

\renewcommand{\P}{\mathcal{P}}
\newcommand{\I}{\mathcal{I}}
\newcommand{\B}{\mathcal{B}}
\renewcommand{\H}{\mathcal{H}}
\renewcommand{\L}{\mathcal{L}}
\newcommand{\M}{\mathcal{M}}

\newcommand{\MA}{multi-agent}

\DeclareMathOperator*{\argmax}{arg\,max}

\begin{document}
	
\title{Multi-Agent Submodular Optimization}
\author{Richard Santiago \and F. Bruce Shepherd}

\institute{Richard Santiago \at
	McGill University, Montreal, Canada \\
	\email{richard.santiagotorres@mail.mcgill.ca}           
	\and
	F. Bruce Shepherd \at
	University of British Columbia, Vancouver, Canada \\
	\email{fbrucesh@cs.ubc.ca}
}

\date{Received: date / Accepted: date}

\maketitle


\begin{abstract}
Recent years have seen many algorithmic advances in the area of submodular optimization: (SO) $\min/\max~f(S): S \in \F$, where $\F$ is a given family of feasible sets over a ground set $V$ and $f:2^V \rightarrow \R$ is submodular. This progress has been coupled with a wealth of new applications for these models. Our focus is on a more general class of {\em multi-agent submodular optimization} (MASO) $\min/\max \sum_{i=1}^{k} f_i(S_i):  S_1 \uplus S_2 \uplus \cdots \uplus S_k \in \F$. Here we use $\uplus$ to denote disjoint union and hence this model is attractive where  resources are being allocated across  $k$ agents, each with its own submodular cost function $f_i()$.  This was introduced in the minimization setting by Goel et al. In this paper we explore the extent to which the approximability of the multi-agent problems are linked to their single-agent versions, referred to informally as the {\em multi-agent gap}.

We present different reductions that transform a multi-agent problem into a single-agent one. For minimization, we show that (MASO) has an $O(\alpha \cdot \min \{k, \log^2 (n)\})$-approximation whenever (SO) admits an $\alpha$-approximation over the convex formulation. In addition, we discuss the class of ``bounded blocker'' families where there is a provably tight $O(\log n)$ multi-agent gap between (MASO) and (SO).
For maximization, we show that monotone (resp. nonmonotone) (MASO) admits an $\alpha (1-1/e)$ (resp. $\alpha \cdot 0.385$) approximation whenever monotone (resp. nonmonotone) (SO) admits an $\alpha$-approximation over the multilinear formulation; and the $1-1/e$ multi-agent gap for monotone objectives is tight. We also discuss several families (such as spanning trees, matroids, and $p$-systems) that have an (optimal) multi-agent gap of 1. These results substantially expand the family of tractable models for submodular maximization.

\keywords{submodular optimization \and approximation algorithms \and multi-agent}

\end{abstract}




\section{Introduction}
\label{sec:intro}

A function $f:2^V \to \R$ is \emph{submodular} if $f(S) + f(T) \geq f(S \cup T) + f(S \cap T)$ for any $S,T \subseteq V$.
We say that $f$ is \emph{monotone} if $f(S) \leq f(T)$ whenever $S \subseteq T$.
Throughout, all submodular functions are nonnegative, and we usually assume that $f(\emptyset)=0$. Our functions are  given by a \emph{value oracle}, where for a given set $S$ an algorithm can query the oracle to find its value $f(S)$.

For a family of feasible sets $S\in\F$ on a finite ground set $V$ we consider the following broad class of submodular optimization (SO) problems:
\begin{equation}
\label{eqn:SA}
\mbox{SO($\F$) ~~~~Min~ / ~Max ~$f(S):S\in\F$}
\end{equation}
where $f$ is a nonnegative submodular set function on $V$. There has been an impressive
recent stream  of activity around these problems for a variety of set families $\F$.
We explore  the connections between these (single-agent) problems and their multi-agent
incarnations.
In the {\em multi-agent (MA)} version, we have $k$ agents  each of which  has an associated nonnegative submodular
set function $f_{i}$, $i \in [k]$. As before, we are looking for  sets $S\in\F$, however,
we now have a 2-phase task: the elements of $S$ must also be partitioned amongst the agents.
Hence we have set variables $S_{i}$ and seek to optimize $\sum_{i}f_{i}(S_{i})$.
This leads to the multi-agent submodular optimization (MASO) versions:
\begin{equation}
\label{eqn:MA}
\mbox{MASO($\F$) ~~~~Min~ / ~Max ~$\sum_{i=1}^{k} f_{i}(S_{i}):S_{1}\uplus S_{2}\uplus \cdots\uplus S_{k}\in\F$.}
\end{equation}

The special case when $\F=\{V\}$ has been  previously examined both for
minimization  (the minimum submodular cost allocation problem \cite{hayrapetyan2005network,svitkina2010facility,ene2014hardness,chekuri2011submodular}) and
maximization (submodular welfare problem \cite{lehmann2001combinatorial,vondrak2008optimal}). For general families $\F$, however, we are only aware of the development in Goel et al. \cite{goel2009approximability} for the minimization setting.
A natural first question is whether any multi-agent problem could be directly reduced (or encoded) to a single-agent one
over the same ground set $V$. Goel et al. give an explicit example where such a reduction does not exist. More emphatically, they show  that when $\F$ consists of
vertex covers in a graph, the {\em single-agent (SA)} version (i.e., (\ref{eqn:SA})) has a 2-approximation while
the MA version has an inapproximability lower bound of $\Omega(\log n)$.

Our first main
objective is to explain  the extent to which approximability
for multi-agent problems is intrinsically connected to their single-agent
versions, which we also refer to as the {\em primitive} associated with $\F$.
We refer to the {\em multi-agent (MA) gap} as the approximation-factor loss incurred
 by moving to the MA setting.

Our second objective is to extend the multi-agent model and show that in some cases
this larger class remains tractable. Specifically,
we define the
{\em capacitated multi-agent submodular optimization (CMASO) problem} as follows:
\begin{equation}
\label{ma}
\mbox{CMASO$(\F)$}~~~~~~~
\begin{array}{rc}
\max / \min & \sum_{i=1}^{k}f_{i}(S_{i}) \\
\mbox{s.t.} & S_{1}\uplus \cdots\uplus S_{k}\in\F\\
 & S_{i}\in\F_{i}\,,\,\forall i\in[k]
\end{array}
\end{equation}

\noindent
where we are supplied with subfamilies $\mathcal{F}_i$.
Many existing applications  fit into this framework and some of these can be enriched through the added flexibility of the capacitated model.
We illustrate this with concrete examples in Section~\ref{sec:applications}.

Prior work in both the single and multi-agent settings is summarized in Section~\ref{sec:related work}.
We present our main results next.

 \subsection{Our contributions}
\label{sec:reduce}

 We first discuss the minimization side of MASO (i.e. (\ref{eqn:MA})). Here the work of \cite{ene2014hardness} showed that for general nonnegative submodular functions the problem is in fact inapproximable within any multiplicative factor even in the case where $\F=\{V\}$ and $k=3$ (since it is NP-Hard to decide whether the optimal value is zero).
Hence we focus almost completely on nonnegative monotone submodular objectives $f_i$. In fact, even in the single-agent setting with  a  nonnegative monotone submodular function $f$, there exist a number of polynomial hardness results  over fairly simple set families $\F$; examples include minimizing a submodular function subject to
a cardinality constraint \cite{svitkina2011submodular} or over the family of spanning trees \cite{goel2009approximability}.

We show, however, that if the SA primitive for a family $\F$ admits approximation via
a natural convex relaxation (see Appendices~\ref{sec:blocking} and \ref{sec:relaxations}) based on the Lov\'asz extension,
then we may extend this to its multi-agent version with a modest blow-up in the approximation factor.

\begin{theorem}
	\label{thm:klog}
	Suppose there is a (polytime) $\alpha(n)$-approximation for monotone SO($\F$) minimization
	via the blocking convex relaxation. Then there is a (polytime) $O(\alpha(n) \cdot \min\{k, \log^2 (n)\})$-approximation
	for monotone MASO($\F$) minimization.
\end{theorem}

We remark that the $O(\log^2 (n))$ approximation loss due to having multiple agents (i.e the  MA gap) is in the right ballpark, since the vertex cover problem has a factor $2$-approximation for single-agent and a tight $O(\log n)$-approximation for the MA version \cite{goel2009approximability}.

We also discuss how Goel et al's $O(\log n)$-approximation for MA vertex cover is a special case of a more general phenomenon. Their analysis
only relies on the fact that the feasible family (or at least its upwards closure) has a {\em bounded blocker property}.
Given a family $\F$, the {\em blocker} $\B(\F)$ of $\F$
consists of the minimal sets $B$ such that $B \cap F \neq \emptyset$ for each $F \in \F$.
We say that $\B(\F)$ is {\em $\beta$-bounded} if $|B| \leq \beta$ for all $B \in \B(\F)$.

Families with bounded blockers have been previously studied in the SA minimization setting, where the works  \cite{koufogiannakis2013greedy,iyer2014monotone}
show that $\beta$-approximations are always available.
Our next result (combined with these) establishes an $O(\log n)$ MA gap for bounded blocker families, thus improving  the $O(\log^2 (n))$ factor  in Theorem \ref{thm:klog} for general families.
We remark that this $O(\log n)$ MA gap is tight due to examples like vertex covers ($2$-approximation for SA and a tight $O(\log n)$-approximation for MA) or submodular facility location ($1$-approximation for SA and a tight $O(\log n)$-approximation for MA).

\begin{theorem}
\label{thm:beta}
Let $\F$ be a family with a $\beta$-bounded blocker. 
Then there is a randomized  $O(\beta \log n)$-approximation algorithm for monotone $MASO(\F)$ minimization.
\end{theorem}

While our work focuses almost completely on monotone objectives, we show in Section \ref{sec:MAfromSA} that upwards closed families with a bounded blocker remain tractable under some special types of nonmonotone objectives introduced by Chekuri and Ene.

We conclude our minimization work by discussing a class of families which behaves well for MA minimization despite not having a bounded blocker. More specifically, in Section~\ref{sec:rings} we observe that crossing (and ring) families have an MA gap of $O(\log n)$.

\begin{theorem}
There is a tight $\ln (n)$-approximation for monotone MASO($\F$) minimization over crossing families $\F$.
\end{theorem}

We now discuss our contributions for the maximization setting. Our main result here establishes that if the SA primitive for a family $\F$ admits approximation via its multilinear relaxation (see Section \ref{sec:max-SA-MA-formulations}),
then we may extend this to its multi-agent version with a constant factor loss.

\begin{theorem}
	\label{thm:max-MA-gap}
	If there is a (polytime) $\alpha(n)$-approximation for monotone SO($\F$) maximization
	via its multilinear relaxation, then there is a (polytime) $(1-1/e) \cdot \alpha(n)$-approximation
	for monotone MASO($\F$) maximization. Furthermore, given a downwards closed family $\F$,
	if there is a (polytime) $\alpha(n)$-approximation 
	for nonmonotone SO($\F$) maximization
	via its multilinear relaxation, then there is a (polytime) $0.385 \cdot \alpha(n)$-approximation
	for nonmonotone MASO($\F$) maximization.
\end{theorem}

We remark that the $(1-1/e)$ MA gap in the monotone case is tight due to examples like
$\F=\{V\}$, where there is a trivial $1$-approximation for the SA problem and a tight $(1-1/e)$-approximation for the MA version \cite{vondrak2008optimal}.

In Section \ref{sec:MASA} we describe a simple generic reduction that
shows that for some families an (optimal) MA gap of $1$ holds.

\begin{theorem}
	\label{thm:max-invariance}
	Let $\F$ be a matroid, a $p$-matroid intersection, or a $p$-system. Then, if there is a (polytime) $\alpha$-approximation algorithm for monotone (resp. nonmonotone) SO($\F$) maximization,
	there is a (polytime) $\alpha$-approximation algorithm for monotone (resp. nonmonotone) MASO($\F$) maximization.
\end{theorem}



In the setting of CMASO (i.e. (\ref{ma})) our results provide additional modelling flexibility.
They imply that one maintains decent approximations even while adding interesting side  constraints.
For instance, for a monotone maximization instance of CMASO
where $\F$ corresponds to a $p$-matroid intersection and the $\F_i$ are all matroids, our results from Section \ref{sec:MASA}
lead to a $(p+1+\epsilon$)-approximation algorithm.
We believe that these, combined with other results from Section \ref{sec:MASA}, substantially expand the family of tractable models for maximization.

While the impact of this reduction is more for maximization, 
it also has some interesting consequences in the minimization setting.
We discuss in Section \ref{sec:reduction-properties} how some of our results
help explaining why for the family of spanning trees, perfect matchings, and $st$-paths,
the approximations factors revealed in \cite{goel2009approximability} for the monotone minimization problem
are the same for both the single-agent and multi-agent versions.

\subsection{Some applications of (capacitated) multi-agent optimization}
\label{sec:applications}

In this section we present several problems  in the literature
which are special cases of Problem (\ref{eqn:MA}) and the more general Problem (\ref{ma}).
We also indicate how the extra generality of {\sc CMASO} (i.e. (\ref{ma})) gives modelling advantages.
We start with the \underline{maximization} setting.

\begin{example}[The Submodular Welfare Problem] 
	~ The most basic problem in the maximization setting arises when
	we take the feasible space $\F=\{V\}$. This describes  a  well-known model (introduced in \cite{lehmann2001combinatorial}) for allocating goods   to agents, each of which
	has a monotone submodular valuation (utility) function over baskets of goods. This is formulated  as  (\ref{eqn:MA}) by considering
	nonnegative monotone functions $f_{i}$ and $\F=\{V\}$.
	The CMASO framework allows us to incorporate  additional constraints into this problem by defining the families $\F_{i}$ appropriately. For instance, one can impose cardinality constraints on the number of elements that an agent can take, or to only allow agent $i$ to take a set $S_i$ of elements satisfying some bounds $L_i \subseteq S_i \subseteq U_i$.
\end{example}

\begin{example}[The Separable Assignment Problem]
	~ An instance of the Separable Assignment Problem (SAP) consists of $m$ items and $n$ bins. Each bin $j$ has an associated downwards closed collection of feasible sets $\F_j$, and a modular function $v_{j}(i)$ that denotes the value of placing item $i$ in bin $j$. The goal is to choose disjoint feasible sets $S_j \in \F_j$ so as to maximize $\sum_{j=1}^n v_j(S_j)$. This well-studied problem (\cite{fleischer2006sap,goundan2007revisiting,calinescu2011maximizing}) corresponds to a CMASO instance where all the objectives are modular, $\F=2^V$, and the families $\F_i$ are downwards closed.
\end{example}


We next discuss an example where using matroid-capacity constraints $\F_i$  in CMASO is beneficial.

\begin{example}[Recommendation Systems and Matroid Constraints] ~ This is a widely deployed class of problems that  entails
	the targeting of product ads to a mass of (largely unknown) buyers or ``channels''.  In  \cite{chen2016conflict} a ``meta'' problem
	is considered where (known) prospective buyers are recommended to interested sellers. This type of recommendation system incurs additional constraints such as
	(i) bounds on the size of the buyer list provided to  each seller  (e.g., constrained by a seller's  budget) and
	(ii) bounds on how often a buyer appears on a list (to not bombard buyers).    These constraints are modelled as
	a ``$b$-matching''  problem in a bipartite buyer-seller graph $G_B$. They also consider a more sophisticated  model which incorporates ``conflict-aware'' constraints
	on the buyer list for each seller, e.g., no more than one buyer from a household should be recommended to a seller. They model conflicts using extra edges amongst the buyer nodes
and they specify
	 an upper bound on the number of allowed conflict edges induced by a seller's recommendation list.  Heuristics for this (linear-objective) model  \cite{chen2016conflict} are successfully  developed on Ebay data, even though
	the computational problem is shown to be NP-hard. In fact, subsequent work \cite{chen2016group} shows that conflict-aware $b$-matching suffers an inapproximability bound of  $O(n^{1-\epsilon})$.
	 We now propose an alternative model  which admits an $O(1)$-approximation. Moreover, we allow a more general submodular multi-agent objective $\sum_i f_i(B_i)$ where $B_i$ are the buyers recommended to seller $i$.

To formulate this in the CMASO model (\ref{ma}) we consider the same complete buyer-seller bipartite graph from previous work.  We now represent a buyer list $B_i$ as a set of edges $S_i$.
In order that each buyer $v$ is not recommended more than its allowed maximum $b(v)$,
we add the constraint that the number of edges in $F=\cup S_i$ which are incident to buyer node $v$ is at most $b(v)$. The family $\F$ of such sets $F$ forms a partition matroid. Hence the problem can be formulated as:
	\[
	\begin{array}{cc}
	\max & \sum_{i=1}^{k}f_{i}(S_{i})\\
	\mbox{s.t.} &  S_1 \uplus  S_2 \cdots \uplus  S_k \in \F \\
	& S_i \in \F_i, \;\; \forall i \in [k]
	\end{array}
	\]

We now
define $\F_i$ to enforce conflict constraints for the seller as follows. Let $V_i$ denote the edges from seller node $i$ to allowable buyers  for $i$ (possibly $V_i$ is all buyers).
	We may then partition $V_i$ into ``households'' $V_{ij}$. In order to model conflicts, we insist that $S_i$ is allowed to include at   most $1$ element from each $V_{ij}$.
	The resulting family  $\F_i$ is  a partition or laminar matroid.
	Our results imply that this new version  has a polytime $O(1)$-approximation (in the value oracle model). 
\end{example}

\begin{example}[Sensor Placement]
~ The problem of placing sensors and information gathering has been popular in the submodularity literature \cite{krause2007near,krause1973optimal,krause08efficient}.
We are given a set of sensors $V$ and a set of possible locations $\{1,2,\ldots,k\}$ where the sensors can be placed.
There is also a budget constraint restricting the number of sensors that can be deployed.
The goal is to place sensors at some of the locations so as to maximize the total ``informativeness''.
Consider a multi-agent objective function $\sum_{i \in [k]} f_i(S_i)$, where $f_i(S_i)$
measures the informativeness of placing sensors $S_i$ at location $i$. It is then natural to consider a diminishing return (i.e. submodularity) property for the $f_i$'s. 
We can then formulate the problem as MASO($\F$)
where $\F:=\{ S \subseteq V: |S| \leq b \}$ imposes the budget constraint. We can also use CMASO for additional modelling flexibility.
For instance, we may define  $\F_i=\{ S\subseteq V_i: |S| \leq b_i \}$
where $V_i$ are the allowed sensors for location $i$ and $b_i$ is an upper bound on the sensors
located there.
\end{example}

We now discuss Problem (\ref{eqn:MA}) and (\ref{ma}) in the \underline{minimization} setting.

\begin{example}[Minimum Submodular Cost Allocation]
\label{ex:MSCA}
~ The most basic problem in the minimization setting arises when
we simply take $\F=\{V\}$.
This problem, $\min \sum_{i=1}^{k}f_{i}(S_{i}): S_{1}\uplus \cdots\uplus S_{k}=V$,
has been widely considered in the literature for both monotone \cite{svitkina2010facility} and
nonmonotone functions \cite{chekuri2011submodular,ene2014hardness},
and is referred to as the {\sc Minimum Submodular Cost Allocation (MSCA) problem}\footnote{Sometimes referred to as submodular procurement auctions.} (introduced in \cite{hayrapetyan2005network,svitkina2010facility}  and further developed in \cite{chekuri2011submodular}).
This is formulated  as  (\ref{eqn:MA}) by taking $\F=\{V\}$. 
The CMASO framework allows us to incorporate  additional  constraints into this problem. The most natural are to impose cardinality constraints on the number of elements that an agent can take, or to only allow agent $i$ to take a set $S_i$ of elements satisfying some bounds $L_i \subseteq S_i \subseteq U_i$. 
\end{example}

\begin{example}[Multi-agent Minimization]
~ Goel et al \cite{goel2009approximability} consider the  special cases
of MASO($\F$) where the objectives are nonnegative monotone submodular and $\F$ is either the
family of vertex covers, spanning trees, perfect matchings, or
shortest $st$ paths.
\end{example}

\subsection{Related work}
\label{sec:related work}

{\bf Single Agent Optimization.} The high level view of the tractability status for unconstrained (i.e., $\F=2^V$)
submodular optimization is that   both maximization and minimization generally behave well. Minimizing a submodular
set function is a classical combinatorial optimization problem which
can be solved in polytime \cite{grotschel2012geometric,schrijver2000combinatorial,iwata2001combinatorial}.
Unconstrained maximization on the
other hand is known to be inapproximable for general submodular set functions
but admits a polytime constant-factor approximation algorithm
when $f$ is nonnegative \cite{buchbinder2015tight,feige2011maximizing}.

In the constrained maximization setting, the classical work  \cite{nemhauser1978analysis,nemhauser1978best,fisher1978analysis} already
established an optimal $(1-1/e)$-approximation
factor for maximizing a nonnegative monotone submodular function subject to a
cardinality constraint, and a $(1/(k+1))$-approximation for maximizing
a nonnegative monotone submodular function subject to $k$ matroid constraints. This approximation is almost tight in the sense that there is an (almost matching)
factor $\Omega(\log(k)/k)$ inapproximability result \cite{hazan2006complexity}.
For nonnegative monotone
functions, \cite{vondrak2008optimal,calinescu2011maximizing} give
 an optimal $(1-1/e)$-approximation based on multilinear extensions when $\F$ is a matroid;
 \cite{kulik2009maximizing} provides a
$(1-1/e-\epsilon)$-approximation when $\F$ is given by a constant number
of knapsack constraints, and \cite{lee2010submodular}
gives a local-search algorithm that achieves a $(1/k-\epsilon)$-approximation
(for any fixed $\epsilon>0$) when $\F$ is a $k$-matroid intersection. For nonnegative nonmonotone functions, a $0.385$-approximation is the best factor known \cite{buchbinder2016constrained} for maximization under a matroid constraint, in \cite{lee2009non} a $1/(k+O(1))$-approximation is given for $k$ matroid constraints with $k$ fixed. A simple ``multi-greedy'' algorithm \cite{gupta2010constrained} matches the approximation of Lee et al. but is polytime for any $k$. Vondrak \cite{vondrak2013symmetry} gives a $\frac{1}{2}(1-\frac{1}{\nu})$-approximation
under a matroid base constraint where $\nu$ denotes the fractional
base packing number. Finally, Chekuri et al \cite{vondrak2011submodular} introduce a general framework based on relaxation-and-rounding that allows for combining different types of constraints. This leads, for instance, to $0.38/k$ and $0.19/k$ approximations for maximizing nonnegative submodular monotone and nonmonotone functions respectively under the combination of $k$ matroids and $\ell=O(1)$ knapsacks constraints.

For constrained minimization, the news is worse
\cite{goel2009approximability,svitkina2011submodular,iwata2009submodular}.
If $\F$ consists of spanning trees (bases of
a graphic matroid) Goel et al \cite{goel2009approximability} show a lower bound
 of $\Omega(n)$, while in the case where $\F$ corresponds
to the cardinality constraint $\{S:|S|\geq k\}$ Svitkina and Fleischer
\cite{svitkina2011submodular} show a lower bound  of $\tilde{\Omega}(\sqrt{n})$.
There are a few exceptions. The problem can be solved exactly when $\F$ is a ring family (\cite{schrijver2000combinatorial}), triple family (\cite{grotschel2012geometric}), or parity family (\cite{goemans1995minimizing}).
In the context of NP-Hard problems, there are almost no cases where good (say $O(1)$ or $O(\log n)$) approximations exist.
We have that the submodular vertex cover admits
a $2$-approximation (\cite{goel2009approximability,iwata2009submodular}), and the $k$-uniform hitting
set has $O(k)$-approximation.

{\bf Multi-agent Problems.} In the  maximization setting the main  multi-agent problem studied is the Submodular Welfare Maximization ($\F=\{V\}$)
 for which the initial  $1/2$-approximation \cite{lehmann2001combinatorial} 
 was  improved to $1-1/e$ by Vondrak \cite{vondrak2008optimal}
who introduced the continuous greedy algorithm. This approximation
is in fact optimal  \cite{khot2005inapproximability,mirrokni2008tight}.
We are not aware of maximization work for Problem (\ref{eqn:MA}) for a nontrivial family $\F$.

For the multi-agent minimization setting, MSCA (i.e. $\F=\{V\}$) is the most studied application of Problem (\ref{eqn:MA}).
For nonnegative monotone functions, MSCA is equivalent
to the Submodular Facility Location problem considered
in \cite{svitkina2010facility}, where a tight
$O(\log n)$ approximation is given.
If  the functions $f_{i}$ are nonnegative
and nonmonotone, then   no
multiplicative factor approximation exists \cite{ene2014hardness}. If, however, the functions
can be written as $f_{i}=g_{i}+h$ for some nonnegative monotone submodular $g_{i}$
and a nonnegative symmetric submodular function $h$, an $O(\log n)$
approximation is given in \cite{chekuri2011submodular}. In the more general case where $h$ is nonnegative submodular, an $O(k \log n)$ approximation is provided in \cite{ene2014hardness}, and this is tight \cite{mirzakhani2014sperner}.

Goel et al \cite{goel2009approximability} consider the minimization case
of  (\ref{ma}) for nonnegative monotone submodular functions, in which $\F$ is a
nontrivial collection of subsets of $V$ (i.e. $\F\subset2^{V}$)
and there is no restriction on the $\F_{i}$ (i.e. $\F_{i}=2^{V}$
for all $i$). In particular, given a graph $G$ they consider the
families of vertex covers, spanning trees, perfect matchings, and
shortest $st$ paths. They provide a tight $O(\log n)$ approximation
for the vertex cover problem, and show polynomial hardness for the
other cases. To the best of our knowledge \cite{goel2009approximability}  is the only work
 on
Problem (\ref{eqn:MA})  for nontrivial collections $\F$.  

\section{Multi-agent submodular minimization}	
\label{sec:multimin}

In this section we seek generic reductions for multi-agent minimization problems to their single-agent primitives. We mainly focus on the case of nonnegative monotone submodular objective functions and we work
with a natural convex relaxation that is obtained via the Lov\'asz extension of a set function (cf. Appendices \ref{sec:blocking} and \ref{sec:relaxations}). We show that if the SA primitive admits approximation via such relaxation, then we may extend this to its MA version up to an $O(\min\{k, \log^2 (n)\})$ factor loss.

As noted already, the $O(\log^2 (n))$ approximation factor loss due to having multiple agents is in the right ballpark since for vertex covers there is a factor $2$-approximation for SA submodular minimization,
and a tight $O(\log n)$-approximation for the multi-agent version \cite{goel2009approximability}.
In Section \ref{sec:MAfromSA} we discuss an extension  of this vertex cover result to a larger class of families with a MA gap of $O(\log n)$.



\subsection{The single-agent and multi-agent formulations}
\label{sec:SA-MA-formulations}

Due to monotonicity, one may often assume that we are working with a
family $\F$ which is {\em upwards-closed} (sometimes referred to as {\em blocking families}), i.e. if $F \subseteq F'$ and $F \in \F$, then $F' \in \F$. 
This can be done without loss of generality even if we seek polytime algorithms, since separation over a polytope with vertices $\{\chi^F: F \in \F\}$ implies separation over its dominant. We refer the reader to Appendix~\ref{sec:blocking} for details.

For a set function $f:\{0,1\}^V \to \R$ with $f(\emptyset)=0$ one can define
its {\em Lov\'asz extension} $f^L:\R_+^V \to \R$ (introduced
in \cite{lovasz1983submodular}) as follows.
Let  $0   <  v_1 < v_2 < ... < v_m$  be the distinct
positive values taken in some vector $z \in \R_+^V$, and let $v_0=0$.
For each $i \in \{0,1,...,m\}$ define the set $S_i:=\{ j:  z_j > v_i\}$. In particular, $S_0$ is the support of $z$
and  $S_m=\emptyset$. One then defines (see Appendix \ref{sec:LE-def} for equivalent definitions):
\[
f^L(z) =  \sum_{i=0}^{m-1}   (v_{i+1}-v_i) f(S_i).
\]

It follows from the definition that $f^L$ is positively homogeneous, that is $f^L(\alpha z)=\alpha f^L(z)$ for any $\alpha > 0$ and $z \in \R_+^V$.
Moreover, it is also straightforward to see that $f^L$ is a monotone function if $f$ is.
We have the following result due to Lov\'asz.
\begin{lemma}
	[Lov\'asz \cite{lovasz1983submodular}]
	The function $f^L$ is convex if and only if $f$ is submodular.
\end{lemma}

This now gives rise to natural convex relaxations for the single-agent and multi-agent problems (see Appendix \ref{sec:relaxations}) based on some upwards closed relaxation $\{z \geq 0:Az\geq r\}$ of the integral polyhedron $conv(\{\chi^S:S\in \F\})$. In particular, let us denote $P(\F):=\{z \geq 0:Az\geq r\}$, and assume $A$ is a matrix with nonnegative integral entries and $r$ is a vector with positive integral components (if $r_i = 0$ then the $ith$ constraint is always satisfied and we can remove it). For simplicity, we also assume that the entries of $A$ are polynomially bounded in $n$. 

The {\em single-agent Lov\'asz extension formulation} (used in \cite{iwata2009submodular,iyer2014monotone}) is:
\begin{equation}
\label{SA-LE}
(\mbox{SA-LE}) ~~~~
\min f^L(z): z \in P(\F),
\end{equation}
and the {\em multi-agent Lov\'asz extension formulation} (used in \cite{chekuri2011submodular} for $\F=\{V\}$) is:
\begin{equation}
\label{MA-LE}
(\mbox{MA-LE}) ~~~~
\min \sum_{i=1}^k f^L_i(z_i): z_1 + z_2 + \cdots + z_k \in P(\F).
\end{equation}

By standard methods (see Appendix \ref{sec:relaxations}) one may solve these problems in polytime
if one can separate over the relaxation $P(\F)$. This is often the case for many
natural families such as spanning trees, perfect matchings, $st$-paths, and vertex covers.

\subsection{Rounding the (MA-LE) formulation for upwards closed families $\F$}
It is shown in \cite{chekuri2011submodular} that in the setting of monotone objectives and $\F = \{V\}$, a fractional solution
of (MA-LE) can be rounded into an integral one at an $O(\log n)$ factor loss.
Moreover, they show this still holds for some special types of nonmonotone objectives.

\begin{theorem}[\cite{chekuri2011submodular}]
	\label{thm:monot-MSCA}
	Let $z_1+z_2+\cdots+z_k$ be a feasible solution for (MA-LE) in the setting where $\F=\{V\}$ (i.e. $\sum_{i\in [k]} z_i = \chi^V$)
	and $f_i = g_i + h$ where the $g_i$ are nonnegative monotone submodular and $h$ is nonnegative symmetric submodular. 
	Then there is a randomized rounding procedure that outputs an integral feasible solution
	$\bar{z}_1+\bar{z}_2+\cdots+\bar{z}_k$ such that $\sum_{i\in [k]} f^L_i (\bar{z}_i) \leq O(\log n) \sum_{i \in [k]} f_i^L(z_i)$ on expectation.
	That is, we get a partition $S_1,S_2,\ldots,S_k$ of $V$ such that $\sum_{i \in [k]}f_i(S_i) \leq O(\log n) \sum_{i \in [k]} f_i^L(z_i)$
	on expectation.
\end{theorem}

Our next result shows that the above rounding procedure can be adapted in a straightforward way to the setting where we have a general upwards closed family $\F$.
We omit the proof to Appendix \ref{sec:nonmonotone}.
\begin{theorem}
	\label{thm:sym-MSCA}
	Consider an instance of (MA-LE) where $\F$ is an upwards closed family and $f_i = g_i + h$ where the $g_i$ are nonnegative monotone submodular and $h$ is nonnegative symmetric submodular. Let $z_1+z_2+\cdots+z_k$ be a feasible solution such that $\sum_{i \in [k]} z_i \geq \chi^U$ for some $U \in \F$.
	Then there is a randomized rounding procedure that outputs an integral feasible solution
	$\bar{z}_1+\bar{z}_2+\cdots+\bar{z}_k$ such that $\sum_{i \in [k]} \bar{z}_i \geq \chi^U$ and $\sum_{i\in [k]} f^L_i (\bar{z}_i) \leq O(\log |U|) \sum_{i \in [k]} f_i^L(z_i)$ on expectation. That is, we get a subpartition $S_1,S_2,\ldots,S_k$ such that $\biguplus_{i \in [k]} S_i \supseteq U \in \F$ and $\sum_{i \in [k]}f_i(S_i) \leq O(\log |U|) \sum_{i \in [k]} f_i^L(z_i)$
	on expectation.
\end{theorem} 

\subsection{A multi-agent gap of $O(\min\{k, \log^2 (n)\})$}
\label{sec:generic-min-approx}

In this section we present the proof of Theorem \ref{thm:klog}.
The main idea behind our reductions is the following. We start with an optimal solution
$z^* = z_1^* + z_2^* + \cdots + z^*_k$ to the multi-agent relaxation (MA-LE) and
build a new feasible solution $\hat{z} = \hat{z}_1 + \hat{z}_2 + \cdots + \hat{z}_k$
where the $\hat{z}_i$ have supports $V_i$ that are pairwise disjoint.
We interpret the $V_i$ as the set of items associated (or pre-assigned) to agent $i$. 
Once we have such a pre-assignment we consider the single-agent problem $\min g(S): S \in \F$ where
\begin{equation}
\label{g-function}
g(S)=\sum_{i=1}^k f_i(S \cap V_i).
\end{equation}
It is clear that $g$ is nonnegative
monotone submodular since the $f_i$ are as well. 
Moreover, for any solution $S \in \F$ for this single-agent problem we obtain
a MA solution of the same cost by setting
$S_i = S \cap V_i$, since  we then have
$g(S) = \sum_{i\in [k]} f_i (S \cap V_i) = \sum_{i \in [k]} f_i(S_i)$.

For a set $S \subseteq V$ and a vector $z \in [0,1]^V $ we denote by $z|_S$ the truncation of $z$ to elements of $S$. 
That is, we set $z|_S (v) = z(v)$ for each $v \in S$ and to zero otherwise. Then notice that by definition of $g$ 
we have that $g^L(z) = \sum_{i \in [k]} f^L_i(z|_{V_i})$. Moreover, if we also have that the $V_i$ are pairwise
disjoint, then $\sum_{i \in [k]} f^L_i(z|_{V_i}) = \sum_{i \in [k]} f^L_i(z_i)$. 
We formalize this observation in the following result.

\begin{proposition}
	\label{prop:g-function}
	Let $z = z_1 + z_2 + \cdots + z_k$ be a feasible solution to (MA-LE) where the
	vectors $z_i$ have pairwise disjoint supports $V_i$. Then
	$g^L(z) = \sum_{i \in [k]} f^L_i(z|_{V_i}) = \sum_{i \in [k]} f^L_i(z_i).$
\end{proposition}

The next two results show how one can get a feasible solution $\hat{z} = \hat{z}_1 + \hat{z}_2 + \cdots + \hat{z}_k$ where the
$\hat{z}_i$ have pairwise disjoint supports, by losing a factor of $O( \log^2 (n) )$ and $k$ respectively.
We remark that these two results combined prove Theorem \ref{thm:klog}.

\begin{theorem}
	\label{thm:min-log}
	Suppose there is a (polytime) $\alpha(n)$-approximation for monotone SO($\F$) minimization
	based on rounding (SA-LE). Then there is a (polytime) $O(\alpha(n) \log (n) \log (\frac{n}{\log n}) )$-approximation
	for monotone MASO($\F$) minimization.
\end{theorem}
\begin{proof}

Let $z^* = z_1^* + z_2^* + \cdots + z^*_k$ denote an optimal solution to (MA-LE) with value $OPT_{frac}$.
In order to apply a black box single-agent rounding algorithm we must create a different multi-agent
solution. This is done in several steps, the first few of which are standard. The key steps
are the {\em fracture, expand and return} steps which arise later in the process.

Let $a_{max}$ denote the largest entry of the matrix $A$. 
Call an element $v$ {\em small} if $z^*(v) \leq \frac{1}{2 n \cdot a_{max}}$.
Then note that the total contribution of small elements in any given constraint is at most a half, i.e. for any row $a_i$ of the matrix $A$ we have $a_i \cdot z|_{small} \leq \frac{1}{2}$.
We obtain a new feasible solution $z' = z'_1 + z'_2 + \cdots + z'_k$ by removing all small elements from the support of the $z^*_i$ and then doubling the resulting vectors. Notice that this is indeed feasible since $A z' \geq 2(r-\frac{1}{2} \cdot \mathbf{1}) = 2r - \mathbf{1} \geq r$, where $\mathbf{1}$ denotes the vector of all ones.
Moreover, by monotonicity and homogeneity of the $f^L_i$, this at most doubles the cost of $OPT_{frac}$.

We now prune the solution $z' = z'_1 + z'_2 + \cdots + z'_k$ a bit more.
Let $Z_j$ be the elements $v$ such that   $z'(v) \in (2^{-(j+1)},2^{-j}]$ for $j=0,1,2,\ldots,L$. Since $z'(v) > \frac{1}{2n\cdot a_{max}}$ for any element
in the support, and we assume that $a_{max}$ is polynomially bounded in $n$, we have that  $L = O(\log n)$. 
We  call $Z_j$ {\em bin $j$} and define $r_j = 2^j$.
We round up each $v \in Z_j$ so that  $z'(v)=2^{-j}$ by augmenting the $z'_i$ values by at most a factor of $2$.
We may do this simultaneously for all $v$ by possibly ``truncating'' the values associated to some of the elements. As before,
this is fine since the $f^L_i$ are monotone.
In the end, we call this a  {\em uniform solution} $z'' = z''_1 + z''_2 + \cdots + z''_k$ in the sense that each $z''(v)$ is some power of $2$. Note that its cost is at most $4 \cdot OPT_{frac}$.

{\sc Fracture.} We now {\em fracture} the vectors $z''_i$ by defining vectors $z''_{i,j} = z''_i |_{Z_j}$ for each $i \in [k]$ and each $j \in \{0,1,\ldots,L\}$, 
where recall that the notation $z|_S$ denotes the truncation of $z$ to elements of $S$. Notice that $z''_i = \sum_{j=0}^L z''_{i,j}$.

{\sc Expand.}  Now for each $j \in \{0,1,\ldots,L\}$ we blow up the vectors $z''_{i,j}$ by a factor $r_j$.  (Don't worry, this scaling is temporary.) Since $z''(v) = \frac{1}{r_j}$ for each $v \in Z_j$, this means that the resulting
values yield a (probably fractional)
cover of $Z_j$. We can then use the rounding procedure discussed in Theorem \ref{thm:monot-MSCA} (with ground set $Z_j$ and taking $h\equiv 0$) to get an integral solution $z'''_{i,j}$ such that $\sum_i f^L_i (z'''_{i,j}) \leq O(\log |Z_j|) \sum_i f^L_i (r_j \cdot z''_{i,j})$ on expectation.

{\sc Return.}  Now we go back to get a new MA-LE solution $\hat{z} = \hat{z}_1 + \hat{z}_2 + \cdots + \hat{z}_k$ by setting $\hat{z}_i = \sum_{j=0}^L \frac{1}{r_j} z'''_{i,j}$. Note that $\hat{z}=z''$ and so this is indeed feasible  (and again uniform). Moreover, we have that the cost of this new solution satisfies
\begin{multline*}
	 \sum_{i=1}^k f^L_i (\hat{z}_i)  =   \sum_{i=1}^k f^L_i ( \sum_{j=0}^L \frac{1}{r_j} z'''_{i,j} ) 
	\leq \sum_{i=1}^k \sum_{j=0}^L \frac{1}{r_j} f^L_i ( z'''_{i,j} ) 
	= \sum_{j=0}^L \frac{1}{r_j} \sum_{i=1}^k f^L_i ( z'''_{i,j} ) \\
	 \leq O( \sum_{j=0}^L \sum_{i=1}^k \log (|Z_j|) f^L_i (z''_{i,j}) )
	\leq O( \sum_{j=0}^L \log (|Z_j|) \sum_{i=1}^k f^L_i (z''_i) )
	 \leq   O( L \cdot \log (\frac{n}{L})) \cdot OPT_{MA} ,
\end{multline*}
where in the first inequality we use the convexity and homogeneity of the $f^L_i$, in the second inequality we use again the homogeneity together with the upper bound for $\sum_i f^L_i (z'''_{i,j})$, in the third inequality we use monotonicity and the fact that $z''_{i,j} \leq z''_i$ for all $j$, and in the last one we use that $\sum_{i=1}^k f^L_i (z''_i) \leq 4 \cdot OPT_{frac} \leq 4 \cdot OPT_{MA}$ and
\begin{equation*}
\sum_{j=0}^L \log |Z_j|  =  \log (\prod_{j=0}^L |Z_j|) \leq  \log (\frac{\sum_{j=0}^L |Z_j|}{L+1})^{L+1} 
=  (L+1) \cdot \log (\frac{n}{L+1}) = O( L \cdot \log (\frac{n}{L})),
\end{equation*}
where the inequality follows from the AM-GM inequality.

{\sc Single-Agent Rounding.}
In the last step we use the function $g$ defined in (\ref{g-function}), with sets $V_i$ corresponding to the support of the $\hat{z}_i$.
Given our $\alpha$-approximation rounding assumption for (SA-LE), we can round $\hat{z}$ 
to find a set $\hat{S}$ such that $g(\hat{S})\leq \alpha g^L(\hat{z})$.
Then, by setting $\hat{S}_i = \hat{S} \cap V_i$ we obtain a MA solution satisfying:
\begin{equation*}
\sum_{i=1}^k f_i(\hat{S}_i) = g(\hat{S}) \leq \alpha g^L(\hat{z}) = \alpha \sum_{i=1}^k f^L_i (\hat{z}_i) \leq \alpha \cdot O( L \cdot \log (\frac{n}{L})) \cdot OPT_{MA},
\end{equation*}
where the second equality follows from Proposition \ref{prop:g-function}.
Since $L = O( \log n )$, this completes the proof.
\qed \end{proof}

We now give an approximation in terms of the number of agents, which becomes preferable when $k < \log^2 (n)$.

\begin{lemma}
	\label{lem:min-k}
	Suppose there is a (polytime) $\alpha(n)$-approximation for monotone SO($\F$) minimization 
	based on rounding (SA-LE). Then there is a (polytime) $k \alpha(n)$-approximation
	for monotone MASO($\F$) minimization.
\end{lemma}
\begin{proof}

	Let $z^* = z_1^* + z_2^* + \cdots + z^*_k$ denote an optimal solution to (MA-LE) with value $OPT_{frac}$.
	We build a new feasible solution $\hat{z} = \hat{z}_1 + \hat{z}_2 + \cdots + \hat{z}_k$ as follows.
	For each element $v \in V$ let $i' = \argmax_{i \in [k]} z^*_i(v)$, breaking ties arbitrarily. 
	Then set $\hat{z}_{i'}(v)=k z^*_i(v)$ and $\hat{z}_{i}(v)=0$ for each $i\neq i'$.
	By construction we have $\hat{z} \geq z^*$, and hence this is indeed a feasible solution. Moreover, by construction we also
	have that $\hat{z}_i \leq k z_i^*$ for each $i \in [k]$. Hence, given the monotonicity and homogeneity of the $f^L_i$ we have	
	\begin{equation*}
		\sum_{i=1}^k f^L_i(\hat{z}_i) \leq \sum_{i=1}^k f^L_i(k z^*_i) = k \sum_{i=1}^k f^L_i(z^*_i) = k \cdot OPT_{frac} \leq k \cdot OPT_{MA}.
	\end{equation*}
	Since the $\hat{z}_i$ have disjoint supports $V_i$, we can now use the function $g$ defined in (\ref{g-function})
	and do a single-rounding argument as in Theorem \ref{thm:min-log}. This completes the proof.
\qed	
\end{proof}

The above lemma has interesting consequences
in the case where $\F=\{V\}$. This is
the submodular facility location
problem considered by Svitkina and Tardos in \cite{svitkina2010facility}. They give
an $O(\log n)$-approximation where $n$   denotes the number of customers/clients/demands.
Lemma \ref{lem:min-k} implies we also have a $k$-approximation
which is  preferable in
facility location problems  where the number of customers swamps the number
of facility locations  (for instance, for Amazon).

\begin{corollary}
	\label{Cor facility-location}
	There is a polytime $k$-approximation for submodular facility location, where
	$k$ denotes the number of facilities.
\end{corollary}
\begin{proof}
	The single-agent version of the problem is the trivial $\min f(S): S \in \{V\}$. Hence a polytime exact algorithm is available for the single-agent problem and thus by Lemma \ref{lem:min-k} a polytime $k$-approximation is available for the multi-agent version.
\qed \end{proof}

\subsection{A tight multi-agent gap of $O(\log n)$ for bounded blocker families}
\label{sec:MAfromSA}

In Section \ref{sec:generic-min-approx} we established an $O( \log^2 (n) )$ MA gap
whenever there is a SA approximation algorithm based on the (SA-LE) formulation.
For the vertex cover problem, however, there is
an improved  MA gap of $O(\log n)$ due to Goel et al. In this section
we generalize their result by describing a larger class of families wtih such MA gap.

Recall that due to monotonicity,
one may often assume that we are working with a
family $\F$ which is {\em upwards-closed}, aka a {\em blocking family} (cf. \cite{iyer2014monotone}).
The advantage is that to certify whether $F \in \F$, we only need to check that $F \cap B \neq \emptyset$ for each element $B$ of the family $\B(\F)$ of minimal blockers of $\F$. We discuss the details in Appendix \ref{sec:blocking}.
The blocking relaxation for a family $\F$ is then given by $P^*(\F):=\{z \geq 0:  z(B) \geq 1 ~\textit{for all $B \in \B(\F)$} \}$.
In this section we consider the formulations (SA-LE) and (MA-LE) in the special case where the fractional relaxation of the integral polyhedron is given by $P^*(\F)$. 

The $2 \ln (n)$-approximation algorithm of Goel et al. for multi-agent vertex cover
relies only on the fact that the feasible set family has the following {\em bounded blocker property}.
We call a clutter (family of noncomparable sets) $\F$ {\em $\beta$-bounded} if $|F| \leq \beta$ for all $F \in \F$. 
We then say that $\F$ has a  $\beta$-bounded blocker if $|B|\leq \beta$ for each $B \in \B(\F)$.

The main SA minimization result for such families is the following.
\begin{theorem}[\cite{iyer2014monotone,koufogiannakis2013greedy}]
\label{thm:SABBagain}
	Let $\F$ be a family with a $\beta$-bounded
	blocker. Then there is a $\beta$-approximation algorithm for monotone $SO(\F)$ minimization.
	If $P^*(\F)$ has a polytime separation oracle, then this is a polytime algorithm.
\end{theorem}

Our next result establishes an $O(\log n)$ MA gap for families with a bounded blocker.
In fact, while our work has focused on monotone objectives (due to the inapproximabiltiy results for general submodular $f_i$)
the next result extends to some special types of nonmonotone objectives.
These were introduced in  \cite{chekuri2011submodular} and \cite{ene2014hardness},
where a tractable middle-ground is found for the minimum submodular cost allocation problem (where $\F = \{V\}$).
They work with objectives $f_i=g_i+h$ where the $g_i$ are monotone submodular and $h$ is symmetric submodular (in \cite{chekuri2011submodular})
or just submodular (in \cite{ene2014hardness}).

We remark that by taking $h\equiv 0$ (which is symmetric submodular),  we obtain a result for monotone functions.
We note that in this setting we do not need $\F$ to be upwards closed, since
due to monotonicity we can work with the upwards closure of $\F$ without loss of generality as previously discussed on Section
\ref{sec:SA-MA-formulations} (see Appendix \ref{sec:blocking} for further details).
Moreover, as previously pointed out, this $O(\log n)$ MA gap is tight due to examples like vertex covers ($2$-approximation for SA and a tight $O(\log n)$-approximation for MA) or submodular facility location ($1$-approximation for SA and a tight $O(\log n)$-approximation for MA).

\begin{theorem}
	\label{BB-nonmonotone}Let $\F$ be an upwards closed family with a $\beta$-bounded
	blocker. Let the objectives be of the form $f_i = g_i + h$ where each $g_i$ is nonnegative monotone submodular and $h$ is nonnegative symmetric submodular. Then there is a randomized  $O(\beta \log n)$-approximation algorithm for the associated $MASO(\F)$ minimization problem.
	If $P^*(\F)$ has a polytime separation oracle, then this is a polytime algorithm.
\end{theorem}
\begin{proof}
	Let $z^* = \sum_{i \in [k]} z^*_i$ be an optimal solution to (MA-LE) based on the blocking relaxation $P^*(\F)$ with value $OPT_{frac}$.
	Consider the new feasible solution given by $\beta z^* = \sum_{i \in [k]} \beta z^*_i$ and let	$U=\{v\in V: \beta z^*(v) \geq 1\}$.
	Since $\F$ has a $\beta$-bounded blocker it follows that $U\in\F$.
	We now have that $\sum_{i \in [k]} \beta z^*_i$ is a feasible solution such that $\sum_{i \in [k]} \beta z^*_i \geq \chi^U$.
	Thus, we can use Theorem \ref{thm:sym-MSCA} to get an integral feasible solution 
	$\sum_{i \in [k]} \bar{z}_i$ such that $\sum_{i \in [k]} \bar{z}_i \geq \chi^U$ and $\sum_{i\in [k]} f^L_i (\bar{z}_i) \leq O(\log |U|) \sum_{i \in [k]} f_i^L(\beta z^*_i) \leq \beta \cdot O(\log n) \cdot OPT_{frac}$ on expectation.
	\qed
\end{proof}

It is shown in \cite{ene2014hardness} (see their Proposition 10) that given any nonnegative
submodular function $h$, one may define  a nonnegative symmetric submodular function $h'$ such that for any partition $S_1,S_2,\ldots,S_k$
we have $\sum_i h'(S_i) \leq k \sum_i h(S_i)$. This, with  our previous result, yields the following corollary.
\begin{corollary}
	\label{BB-nonmonotone-again}Let $\F$ be an upwards closed family with a $\beta$-bounded
	blocker. Let the objectives be of the form $f_i = g_i + h$ where each $g_i$ is nonnegative monotone submodular and $h$ is nonnegative submodular. Then there is a randomized $O(k \beta \log n)$-approximation algorithm for the associated $MASO(\F)$ minimization problem.
\end{corollary}

One may wish to view the above results through the lens of MA gaps, leading to the question of what is the associated SA primitive?
For $g_i+h$ objectives, the SA version should be for general (or symmetric) nonnegative submodular objectives. Moreover, as we only use upwards closed families, one may deduce that these single-agent versions have $\beta$-appoximations via the concept of the monotone closure of a nonmonotone objective from \cite{iyer2014monotone}. Hence our results establish MA gaps of $O(\log n)$ (resp. $O(k \log n)$)
in these nonmonotone settings (and the factor of $k$ is tight \cite{mirzakhani2014sperner}).

Before concluding this section, we note that Theorems~\ref{thm:klog} and \ref{thm:SABBagain}  imply
a $k \beta$-approximation for families with bounded blockers, which becomes preferable when $k < O(\log n)$.

\begin{corollary}
\label{cor:kgapBB}
Let $\F$ be a family with a $\beta$-bounded
	blocker. Then there is a $k \beta$-approximation algorithm for the associated monotone $MASO(\F)$ minimization problem.
\end{corollary}

\subsection{A tight multi-agent gap of $O(\log n)$ for ring and crossing families}	
\label{sec:rings}

It is well known (\cite{schrijver2000combinatorial}) that submodular minimization can be solved exactly in polynomial time over a ring family. In this section we observe that the MA problem over this type of constraint admits a tight $\ln (n)$-approximation.
More generally, we consider {\em crossing families}.
A family $\F$ of subsets of $V$ forms a ring  family (aka lattice family) if
for each $A,B \in \F$ we also have $A \cap B, A \cup B \in \F$.
A crossing family is one where we only require it for sets where $A \setminus B, B \setminus A, A \cap B, V- (A \cup B)$ are all non-empty. Hence any ring family is a crossing family.

For any crossing family $\F$ and any $u,v \in V$, let $\F_{uv}=\{A \in \F:  u \in A, v \not\in A\}$.
It is easy to see that $\F_{uv}$ is a ring family. Moreover, we may solve the original MA problem by solving the associated MA problem for each non-empty $\F_{uv}$ and then selecting the best output solution.

So we assume now that we are given a ring family in such a way that we may compute its minimal
 set $M$ (which is unique).   This is a standard assumption
when working with ring families (cf. submodular minimization algorithm described in \cite{schrijver2000combinatorial}).
Then, due to monotonicity and the fact that $\F$ is closed under intersections, it is not hard to
see that the original problem reduces to the facility location problem
$$
\min \sum_{i=1}^k f_i(S_i): S_1 \uplus  \cdots \uplus  S_k = M \; ,
$$
which admits a tight $(\ln |M|)$-approximation (\cite{svitkina2010facility}). In particular, for the special case
where we have the trivial ring family $\F = \{V\}$ we get a tight $\ln (n)$-approximation.
The next result summarizes these observations.

\begin{theorem}
There is a tight $\ln (n)$-approximation for monotone $MASO(\F)$ minimization over  crossing families $\F$.
\end{theorem}

\section{Multi-agent submodular maximization}
\label{sec:MASA}

In this section we describe two different reductions. 
The first one reduces the capacitated multi-agent problem (\ref{ma}) to a single-agent problem,
and it is based on the simple idea of taking $k$ disjoint copies of the original ground set. 
We show that several properties of the objective and family of feasible sets stay \emph{invariant} (i.e. preserved) under the reduction. 
We use this to establish an (optimal) MA gap of 1 for several families. Examples of such families include spanning trees, matroids, and $p$-systems.

Our second reduction is based on the multilinear extension of a set function.
We establish that if the SA primitive admits approximation via its multilinear relaxation (see Section \ref{sec:max-SA-MA-formulations}),
then we may extend this to its MA version with a constant factor loss, in the monotone and nonmonotone settings. 
Moreover, for the monotone case our MA gap is tight.


\subsection{The lifting reduction}
\label{sec:lifting-reduction}

In this section we describe a generic reduction of (\ref{ma}) to a single-agent problem
$$\max/\min f(S): S\in\L.$$
The argument is based on the idea of viewing assignments of elements $v$ to agents $i$ in a {\em multi-agent bipartite graph}.
This simple idea (which is equivalent to making $k$ disjoint copies of the ground set) already appeared in the classical work of Fisher et al \cite{fisher1978analysis}, and has since then been widely used \cite{lehmann2001combinatorial,vondrak2008optimal,calinescu2011maximizing,singh2012bisubmodular}.
We review briefly the reduction here for completeness and to fix notation. 

Consider the complete bipartite graph $G=([k]+V,E)$. 
Every subset of edges $S \subseteq E$ can be written uniquely as
$S =  \uplus_{i \in [k]} (\{i\} \times S_i)$ for some sets $S_i \subseteq V$.
This allows us to go from a multi-agent objective (such as the one in (\ref{ma})) to a univariate objective $f:2^{E}\to\R$ over the lifted space. Namely, for each set $S \subseteq E$ we define $f(S)=\sum_{i \in [k]} f_i(S_i)$. The function $f$ is well-defined because each subset $S\subseteq E$ can be uniquely written as $S = \uplus_{i \in [k]} (\{i\} \times S_i)$ for some $S_i \subseteq V$.

We consider two families of sets over $E$ that capture the original constraints:
$$
\F':=\{S\subseteq E:S_{1}\uplus \cdots\uplus  S_{k}\in\F\} \hspace{15pt} \mbox{and} \hspace{15pt}
\H:=\{S\subseteq E:S_{i}\in\F_{i},\;\forall i\in[k]\}.
$$

\noindent
We now have:
\[
\begin{array}{cccccccc}
\max/\min & \sum_{i \in [k]} f_i(S_i) & = & \max/\min & f(S) & = & \max/\min & f(S)\\
\mbox{s.t.} & S_{1}\uplus \cdots\uplus S_{k}\in\F &  & \mbox{s.t.} & S\in\F' \cap \H &  & \mbox{s.t.} & S\in\L\\
 & S_{i}\in\F_{i}\,,\,\forall i\in[k]
\end{array},
\]
where in the last step we just let $\L:=\F' \cap\H$.

This reduction is  interesting if our new function
$f$ and family of sets $\L$ have properties which allows us to handle
them computationally. This will depend on the original structure
of the functions $f_i$ and the set families $\F$ and $\F_{i}$. 
In terms of the objective, it is straightforward
to check (as previously pointed out in \cite{fisher1978analysis}) that if the $f_i$ are (nonnegative, respectively monotone)
submodular functions, then $f$ as defined above is also (nonnegative, respectively monotone) submodular.
In Section \ref{sec:reduction-properties} we discuss several properties of the families $\F$ and $\F_i$ that are preserved under this reduction.

\subsection{The single-agent and multi-agent formulations}
\label{sec:max-SA-MA-formulations}

For a set function $f:\{0,1\}^V \to \R$ we define its \emph{multilinear extension}  $f^M:[0,1]^V \to \R$ (introduced in \cite{calinescu2007maximizing}) as
\begin{equation*}
f^M(z)=\sum_{S \subseteq V} f(S) \prod_{v \in S} z_v \prod_{v \notin S} (1-z_v).
\end{equation*}
An alternative way to define $f^M$ is in terms of expectations. Consider a vector $z \in [0,1]^V$ and let $R^z$ denote a random set that contains element $v_i$ independently with probability $z_{v_i}$. Then $f^M(z)= \E[f(R^z)]$, where the expectation is taken over random sets generated from the probability distribution induced by $z$. 

This gives rise to natural single-agent and multi-agent relaxations.
The {\em single-agent multilinear extension relaxation} is:
\begin{equation}
\label{SA-ME}
(\mbox{SA-ME}) ~~~~
\max f^M(z): z \in P(\F),
\end{equation}
and the {\em multi-agent multilinear extension relaxation} is:
\begin{equation}
\label{MA-ME}
(\mbox{MA-ME}) ~~~~
\max \sum_{i=1}^k f^M_i(z_i): z_1 + z_2 + \cdots + z_k \in P(\F),
\end{equation}
where $P(\F)$ denotes some relaxation of $conv(\{\chi^S:S\in \F\})$. While the relaxation (SA-ME) has been used extensively \cite{calinescu2011maximizing,lee2009non,feldman2011unified,ene2016constrained,buchbinder2016constrained} in the submodular maximization literature, we are not aware of any previous work using the multi-agent relaxation (MA-ME).

The following result shows that when $f$ is nonnegative submodular and the formulation $P(\F)$ is downwards closed and admits a polytime separation oracle, the relaxation (SA-ME) can be solved approximately in polytime.

\begin{theorem}[\cite{buchbinder2016constrained,vondrak2008optimal}]
	\label{thm:multilinear-solve-monot}
	Let $f:2^V \to \R_+$ be a nonnegative submodular function and $f^M:[0,1]^V \to \R_+$ its multilinear extension. Let $P \subseteq [0,1]^V$ be any downwards closed polytope that admits a polytime separation oracle, and denote $OPT = \max f^M(z): z\in P$. Then there is a polytime algorithm (\cite{buchbinder2016constrained}) that finds $z^* \in P$ such that $f^M(z^*) \geq 0.385 \cdot OPT$. Moreover, if $f$ is monotone there is a polytime algorithm (\cite{vondrak2008optimal}) that finds $z^* \in P$ such that $f^M(z^*) \geq (1-1/e) OPT$. 
\end{theorem}

For monotone objectives the assumption that $P$ is downwards closed is without loss of generality.
This is not the case, however, when the objective is nonmonotone. Nonetheless, this restriction
is unavoidable, as Vondr{\'a}k \cite{vondrak2013symmetry} showed that no algorithm can find $z^* \in P$ such that $f^M(z^*) \geq c \cdot OPT$
for any constant $c>0$ when $P$ admits a polytime separation oracle but it is not downwards closed.

We can solve (MA-ME) to the same approximation factor as (SA-ME). This follows from the fact that the
MA problem has the form $\{ \max g(w) : w \in W \subseteq {\bf R}^{nk} \}$
where $g(w)=g(z_1,z_2,\ldots,z_k)=\sum_{i \in [k]} f^M_i(z_i)$ and $W$ is the downwards closed polytope $\{w=(z_1,...,z_k): \sum_i z_i \in P(\F)\}$. Clearly we have a polytime separation oracle for $W$ given that we have one for $P(\F)$.
Moreover, it is straightforward to check (see Lemma \ref{lem:max-multilinear} on Appendix \ref{sec:Appendix-Invariance}) that $g(w)=f^M(w)$, where $f$ is the function on the lifted space after applying the lifting reduction from Section \ref{sec:lifting-reduction}. Thus, $g$ is the multilinear extension of a nonnegative submodular function, and we can now use Theorem \ref{thm:multilinear-solve-monot}.

\subsection{A tight multi-agent gap of $1-1/e$}
\label{sec:max-MA-gap}

In this section we present the proof of Theorem \ref{thm:max-MA-gap}.
The high-level idea behind our reduction is the same as in the minimization setting (see Section \ref{sec:generic-min-approx}). That is, we start with an (approximate) optimal solution
$z^* = z_1^* + z_2^* + \cdots + z^*_k$ to the multi-agent (MA-ME) relaxation and
build a new feasible solution $\hat{z} = \hat{z}_1 + \hat{z}_2 + \cdots + \hat{z}_k$
where the $\hat{z}_i$ have supports $V_i$ that are pairwise disjoint.
We then use for the SA rounding step the single-agent problem (as previously defined in (\ref{g-function}) for the minimization setting) $\max g(S): S \in \F$ where
$g(S)=\sum_{i \in [k]} f_i(S \cap V_i)$.

Similarly to Proposition \ref{prop:g-function} which dealt with the Lov\'asz extension,
we have the following result for the multilinear extension.

\begin{proposition}
	\label{prop:max-g-function}
	Let $z = \sum_{i\in [k]} z_i$ be a feasible solution to (MA-ME) where the
	vectors $z_i$ have pairwise disjoint supports $V_i$. Then
	$g^M(z) = \sum_{i \in [k]} f^M_i(z|_{V_i}) = \sum_{i \in [k]} f^M_i(z_i).$
\end{proposition}

We now have all the ingredients to prove our main result in the maximization setting.

\begin{theorem}
	If there is a (polytime) $\alpha(n)$-approximation for monotone SO($\F$) maximization
	based on rounding (SA-ME), then there is a (polytime) $(1-1/e) \cdot \alpha(n)$-approximation
	for monotone MASO($\F$) maximization. Furthermore, given a downwards closed family $\F$,
	if there is a (polytime) $\alpha(n)$-approximation 
	for nonmonotone SO($\F$) maximization
	based on rounding (SA-ME), then there is a (polytime) $0.385 \cdot \alpha(n)$-approximation
	for nonmonotone MASO($\F$) maximization.
\end{theorem}
\begin{proof}
	
	We  discuss first the case of monotone objectives.
	
	{\sc STEP 1.}	
	Let $z^* = z_1^* + z_2^* + \cdots + z^*_k$ denote an approximate solution to (MA-ME) obtained via Theorem \ref{thm:multilinear-solve-monot}, and let $OPT_{frac}$ be the value of an optimal solution. We then have that $\sum_{i \in [k]} f^M_i(z^*_i) \geq (1-1/e) OPT_{frac} \geq (1-1/e) OPT_{MA}$.
	
	{\sc STEP 2.}
	For an element $v \in V$ let $\bf{e_v}$ denote the characteristic vector of $\{v\}$, i.e. the vector in $\R^V$ which has value $1$ in the $v$-th component and zero elsewhere. Notice that by definition of the multilinear extension we have that the functions $f^M_i$ are linear along directions $\bf{e_v}$ for any $v \in V$. It then follows that the function
	\begin{equation*}
	h(t) = f^M_i (z^*_i + t \mathbf{e_v} ) + f^M_{i'} (z^*_{i'} - t \mathbf{e_v} )  +  \sum_{j\in [k], j\neq i,i'} f^M_j(z^*_j)
	\end{equation*}
	is also linear for any $v\in V$ and $i \neq i' \in [k]$, since it is the sum of linear functions (on $t$). In particular, given any $v \in V$ such that there exist $i \neq i' \in [k]$ with $z^*_i(v),z^*_{i'}(v)>0$, there is always a choice so that increasing one component and decreasing the other by the same amount does not decrease the objective value. We use this as follows.
	
	Let $v \in V$ be such that there exist $i \neq i' \in [k]$ with $z^*_i(v),z^*_{i'}(v)>0$. Then, we either set $z^*_i(v) = z^*_i(v) + z^*_{i'}(v)$ and $z^*_{i'}(v) = 0$, or $z^*_{i'}(v) = z^*_i(v) + z^*_{i'}(v)$ and $z^*_i(v) = 0$, whichever does not decrease the objective value. We repeat until the vectors $z^*_i$ have pairwise disjoint support. Let us denote these new vectors by $\hat{z}_i$ and let $\hat{z}= \sum_{i\in [k]} \hat{z}_i$. Then notice that the vector $z^* = \sum_{i\in [k]} z^*_i$ remains invariant after performing each of the above updates (i.e. $\hat{z} = z^*$), and hence the new vectors $\hat{z}_i$ remain feasible.

	{\sc STEP 3.}
	In the last step we use the function $g$ defined in (\ref{g-function}), with sets $V_i$ corresponding to the support of the $\hat{z}_i$.
	Given our $\alpha$-approximation rounding assumption for (SA-ME), we can round $\hat{z}$ 
	to find a set $\hat{S}$ such that $g(\hat{S})\geq \alpha g^M(\hat{z})$.
	Then, by setting $\hat{S}_i = \hat{S} \cap V_i$ we obtain a MA solution satisfying
	\begin{equation*}
	\sum_{i \in [k]} f_i(\hat{S}_i) = g(\hat{S}) \geq \alpha g^M(\hat{z}) = \alpha \sum_i f^M_i (\hat{z}_i) \geq \alpha \sum_i f^M_i (z^*_i) \geq \alpha (1-1/e) OPT_{MA},
	\end{equation*}
	where the second equality follows from Proposition \ref{prop:max-g-function}.
	This completes the proof for monotone objectives.
	
	In the nonmonotone case the proof is very similar. 
	Here we restrict our attention to downwards closed families, since then we can get a $0.385$-approximation at STEP 1 via Theorem \ref{thm:multilinear-solve-monot}. We then apply STEP 2 and 3 in the same fashion as we did for monotone objectives. This leads to a $0.385 \cdot \alpha(n)$-approximation for the multi-agent problem.

\qed \end{proof}

\subsection{Invariance under the lifting reduction}
\label{sec:reduction-properties}

In Section \ref{sec:max-MA-gap} we established a MA gap of $(1-1/e)$ for monotone objectives and of $0.385$ for nonmonotone objectives and downwards closed families based on the multilinear formulations. In this section we describe several families with an (optimal) MA gap of $1$. Examples of these family classes include spanning trees, matroids, and $p$-systems. Moreover, the reduction in this case is completely black box, and hence do not depend on the multilinear (or some other particular) formulation.

We saw in Section \ref{sec:lifting-reduction} how if the original functions $f_i$ are all submodular, then the lifted function $f$ is also submodular. We now focus on the properties of the original families $\F_i$ and $\F$ that are also preserved under the lifting reduction. We show, for instance, that if the family $\F$ induces a matroid (or more generally a $p$-system) over the original ground set $V$, then so does the family $\F'$ over the lifted space $E$.
We summarize these results in Table \ref{table:properties-preserved}, and present most of the proofs in this section.
We next discuss some of the algorithmic consequences of these invariance results.


\begin{table}[H]
	\caption{Invariant properties under the lifting reduction}
	\label{table:properties-preserved}
	\resizebox{\linewidth}{!}{
		\begin{tabular}{|c|c|c|c|}
			\hline
			& \textbf{Multi-agent problem} & \textbf{Single-agent (i.e. reduced) problem} & Result\tabularnewline
			\hline
			1 & $f_i$ submodular & $f$ submodular & \cite{lehmann2001combinatorial} \tabularnewline
			\hline
			2 & $f_i$ monotone & $f$ monotone & \cite{lehmann2001combinatorial} \tabularnewline
			\hline
			3 & $(V,\F)$ a $p$-system & $(E,\F')$ a $p$-system & Prop \ref{prop:p-systems} \tabularnewline
			\hline
			4 & $\F$ = bases of a $p$-system & $\F'$ = bases of a $p$-system & Corollary \ref{cor:p-systems bases} \tabularnewline
			\hline
			5 & $(V,\F)$ a matroid & $(E,\F')$ a matroid & Corollary \ref{cor:matroid-invariant}\tabularnewline
			\hline
			6 & $\F$ = bases of a matroid & $\F'$ = bases of a matroid & Corollary \ref{cor:matroid-bases} \tabularnewline
			\hline
			7 & $(V,\F)$ a $p$-matroid intersection & $(E,\F')$ a $p$-matroid intersection & Appendix \ref{sec:Appendix-Invariance} \tabularnewline
			\hline
			8 & $(V,\F_{i})$ a matroid for all $i\in[k]$ & $(E,\H)$ a matroid & Appendix \ref{sec:Appendix-Invariance} \tabularnewline
			\hline
			9 & $\F_{i}$ a ring family for all $i\in[k]$ & $\H$ a ring family & Appendix \ref{sec:Appendix-Invariance} \tabularnewline
			\hline
			10 & $\F=$ forests (resp. spanning trees) & $\F'=$ forests (resp. spanning trees) &  Section \ref{sec:reduction-properties}\tabularnewline
			\hline
			11 & $\F=$ matchings (resp. perfect matchings) & $\F'=$ matchings (resp. perfect matchings) &  Section \ref{sec:reduction-properties}\tabularnewline
			\hline
			12 & $\F=$ $st$-paths & $\F'=$ $st$-paths &  Section \ref{sec:reduction-properties}\tabularnewline
			\hline
		\end{tabular}
	}
\end{table}


In the setting of MASO (i.e. (\ref{eqn:MA})) this invariance allows us to leverage several results
from the single-agent to the multi-agent setting. These are based on the following result, which uses
the fact that the size of the lifted space $E$ is $nk$.
\begin{theorem}
	\label{thm:max-invariance1}
	Let $\F$ be a matroid, a $p$-matroid intersection, or a $p$-system. If there is a (polytime) $\alpha(n)$-approximation algorithm for monotone (resp. nonmonotone) SO($\F$) maximization (resp. minimization),
	then there is a (polytime) $\alpha(nk)$-approximation algorithm for monotone (resp. nonmonotone) MASO($\F$) maximization (resp. minimization).
\end{theorem}

In both the monotone and nonmonotone maximization settings, the approximation factors $\alpha(n)$ for the family classes described in the theorem above are independent 
of $n$. Hence, we immediately get that  
$\alpha(nk)=\alpha(n)$ for these cases, and thus approximation factors for the corresponding MA problems are the same as for the SA versions.
In our MA gap terminology this implies a MA gap of 1 for such problems.

In the setting of CMASO (i.e. (\ref{ma})) the results described on entries $8$ and $9$ of Table \ref{table:properties-preserved} 
provide additional modelling flexibility. This allows us to combine several constraints while keeping approximation factors fairly good.
For instance, for a monotone maximization instance of CMASO
where $\F$ corresponds to a $p$-matroid intersection and the $\F_i$ are all matroids, the above invariance results
lead to a $(p+1+\epsilon$)-approximation.

We now prove some of the results from Table \ref{table:properties-preserved}.
We start by presenting some definitions that will be useful.
For a subset of edges $S \subseteq E$ we define its \emph{coverage} $cov(S)$ as the set of nodes $v \in V$ saturated by $S$.
That is, $v\in cov(S)$ if there exists $i\in [k]$ such that $(i,v)\in S$.
We then note that by definition of $\F'$ it is straightforward to see that
for each $S \subseteq E$ we have that
\begin{equation}
	\label{def:family-H}
	S \in \F' \iff cov(S) \in \F \mbox{ and } |S|=|cov(S)|.
\end{equation}

For a set $S \subseteq E$, a set $B\subseteq S$ is called a \emph{basis} of $S$ if
$B$ is an inclusion-wise maximal independent subset of $S$. 
Our next result describes how bases and their cardinalities behave under the lifting reduction.

\begin{lemma}
	\label{lem:bases-invariance}
	Let $S$ be an arbitrary subset of $E$. Then for any basis $B$ (over $\F'$) of $S$ there
	exists a basis $B'$ (over $\F$) of $cov(S)$ such that $|B'|=|B|$.
	Moreover, for any basis $B'$ of $cov(S)$ there exists a basis $B$ of $S$ such that $|B|=|B'|$.
\end{lemma}
\begin{proof}
	For the first part, let $B$ be a basis of $S$ and take $B':=cov(B)$. Since $B \in \F'$
	we have by (\ref{def:family-H}) that $B'\in \F$ and $|B'| = |B|$. Now, if $B'$ is not a basis of $cov(S)$ then we can
	find an element $v \in cov(S)-B'$ such that $B'+v \in \F$. 
	Moreover, since $v \in cov(S)$ there exists $i \in [k]$ such that $(i,v)\in S$.
	But then we have that $B+(i,v) \subseteq S$ and
	$B+(i,v)\in \F'$, a contradiction with the fact that $B$ was a basis of $S$.
	
	For the second part, let $B'$ be a basis of $cov(S)$. For each $v \in B'$ let $i_v$ be such that
	$(i_v,v) \in S$, and take $B:=\uplus_{v \in B'} (i_v,v)$. It is clear by definition of $B$ that
	$cov(B)=B'$ and $|B|=|B'|$. Hence $B \in \F'$ by (\ref{def:family-H}). If $B$ is not a basis of $S$
	there exists an edge $(i,v)\in S-B$ such that
	$B+(i,v)\in \F'$. But then by (\ref{def:family-H}) we have that $cov(B+(i,v))\in \F$
	and $B'\subsetneq cov(B+(i,v))\subseteq cov(S)$, a contradiction
	since $B'$ was a basis of $cov(S)$.
\qed 
\end{proof}

We say that $(V,\F)$ is a \emph{$p$-system} if for each $U \subseteq V$, the cardinality
of the largest basis of $U$ is at most $p$ times the cardinality of the smallest basis of $U$.
Our following result is a direct consequence of Lemma \ref{lem:bases-invariance}.

\begin{proposition}
	\label{prop:p-systems}If $(V,\F)$ is a $p$-system, then $(E,\F')$
	is a $p$-system.
\end{proposition}

\begin{corollary}
	\label{cor:p-systems bases}
	If $\F$ corresponds to the set of bases of a $p$-system $(V,\I)$, then $\F'$ also corresponds to the set of bases of some $p$-system $(E,\I')$.
\end{corollary}
\begin{proof}
	Consider $(E,\I')$ where $\I':=\{S\subseteq E: cov(S)\in \I \mbox{ and } |cov(S)|=|S|\}$.
	Then by Proposition \ref{prop:p-systems} we have that $(E,\I')$ is a $p$-system.
	It is now straightforward to check that $\F'$ corresponds precisely to the set of bases of $(E,\I')$.
\qed
\end{proof}

The following two results follow from Proposition \ref{prop:p-systems} and Corollary \ref{cor:p-systems bases} and the fact that matroids are precisely the class of 1-systems.
\begin{corollary}
\label{cor:matroid-invariant}If $(V,\F)$ is a matroid, then $(E,\F')$
is a matroid.
\end{corollary}

\begin{corollary}
	\label{cor:matroid-bases}Assume $\F$ is the set of bases of some matroid
	$\M=(V,\I)$, then $\F'$ is the set of bases of some matroid $\M'=(E,\I')$. 
\end{corollary}

We now focus on families defined over the set of edges of some graph $G$.
To be consistent with our previous notation we denote
by $V$ the set of edges of $G$, since this is the ground set of
the original problem. 
The lifting reduction is based
on the idea of making $k$ disjoint copies for each original element,
and visualize the new ground set (or lifted space) as edges in a bipartite
graph. However, when the original ground set corresponds to the set
of edges of some graph $G$, we may just think of the lifted
space as being the set of edges of the graph $G'$ obtained by taking
$k$ disjoint copies of each original edge. We think of the
edge that corresponds to the $ith$ copy of $v$ as assigning element
$v$ to agent $i$. 
We can formalize this by defining a mapping $\pi:E\to E'$
that takes an edge $(i,v)\in E$ from the lifted space to the edge in $G'$ that corresponds
to the $ith$ copy of the original edge $v$. It is clear that $\pi$
is a bijection.
Moreover, notice that given any graph $G$ and family
$\F$ of forests (as subset of edges) of $G$, the bijection $\pi : E \to E'$ also
satisfies that $\bar{\F}:=\{\pi (S): S \in \F'\}$ is precisely the family of forests of $G'$.
That is, there is a one-to-one correspondence between forests of $G'$ and assignments
$S_1 \uplus S_2 \cdots \uplus S_k = S \in \F$ for the original MA problem where $S$ is a forest of $G$.
The same holds for spanning trees, matchings, perfect matchings, and st-paths.

This observation becomes  algorithmically useful given that $G'$ has 
the same number of nodes of $G$. 
Thus, any approximation factor or hardness result for the above combinatorial structures
that depend on the number of nodes (and not on the number of edges) of the underlying graph,
will remain the same in the MA setting.
We note that this explains why for the family of spanning trees, perfect matchings, and $st$-paths,
the approximations factors revealed in \cite{goel2009approximability} for the monotone minimization problem 
are the same for both the SA and MA versions. Our next result summarizes this.

\begin{theorem}
	\label{thm:max-invariance2}
	Let $\F$ be the family (seen as edges) of forests, spanning trees, matchings, perfect matchings, or $st$-paths of a
	graph $G$ with $m$ nodes and $n$ edges.
	Then, if there is a (polytime) $\alpha(n)$-approximation algorithm for monotone (resp. nonmonotone) SO($\F$) maximization (resp. minimization), there is a (polytime) $\alpha(nk)$-approximation algorithm for monotone (resp. nonmonotone) MASO($\F$) maximization (resp. minimization). Moreover, if there is a (polytime) $\alpha(m)$-approximation algorithm for monotone (resp. nonmonotone) SO($\F$) maximization (resp. minimization),
	then there is a (polytime) $\alpha(m)$-approximation algorithm for monotone (resp. nonmonotone) MASO($\F$) maximization (resp. minimization).
\end{theorem}


\section{Conclusion}
\label{sec:conclusions}

A number of interesting questions remain. Perhaps the main one being whether the $O(\log^2 (n))$ MA gap for minimization can be improved to $O(\log n)$?  We have shown this is the case for bounded blocker and crossing families. Another question is whether  the $\alpha \log^2 (n)$ and $ \alpha k $ approximations can be made truly black box? I.e., do not depend on the convex formulation.

On separate work (\cite{santiago2016multivariate}) we discuss multivariate submodular objectives.  We show that our reductions for maximization
remain well-behaved algorithmically and this opens up more tractable models.  This is the topic of planned future work.


\bibliographystyle{spmpsci}
\bibliography{REFERENCES}

\appendix
\begin{center}
	\large \textbf{APPENDIX}
\end{center}
\section{Upwards-closed (aka blocking) families}
\label{sec:blocking}

In this section, we give some background for blocking families.
As our work for minimization is restricted to  monotone functions, we can often convert an arbitrary
set family into its upwards-closure (i.e., a blocking version of it) and work with it instead. We discuss this
reduction as well.
The technical details discussed in this section are fairly standard and we include them for completeness. Several of these results have already appeared in \cite{iyer2014monotone}.


\subsection{Blocking families and a natural relaxation for the integral polyhedron}
\label{sec:block}

A set family $\F$, over a ground set $V$ is {\em upwards-closed}
if $F \subseteq F'$ and $F \in \F$, implies that $F' \in \F$; these are sometimes referred to as  {\em blocking families}.
Examples of such families include vertex-covers or set covers more generally,  whereas spanning trees are not.

For a blocking family $\F$ one often works
with the induced   sub-family $\F^{min}$ of minimal sets. Then $\F^{min}$ has the property that it is
a {\em clutter}, that is, $\F^{min}$ does not contain
a pair of comparable sets, i.e., sets $F \subset F'$.
 If $\F$ is a clutter, then $\F=\F^{min}$ and there is an associated {\em blocking} clutter $\B(\F)$, which consists of the minimal sets $B$ such that $B \cap F \neq \emptyset$ for each $F \in \F$. We refer to $\B(\F)$ as the {\em blocker} of $\F$.

 One also checks that for an arbitrary upwards-closed family $\F$, we have the following.
\begin{claim}[Lehman]
	\label{claim:blocker}
	\begin{enumerate}
		\item $F \in \F$ if and only if $F \cap B \neq \emptyset$ for all $B \in \B(\F^{min})$.
		\item $\B(\B(\F^{min})) = \F^{min}$.
	\end{enumerate}
\end{claim}

\noindent
Thus the significance of blockers is that one may assert membership in an upwards-closed family $\F$  by checking intersections on sets from the blocker $\B(\F^{min})$. If we  define $\B(\F)$ to be the minimal sets which intersect every element of $\F$, then one checks that $\B(\F)=\B(\F^{min})$.
These  observations lead to   a natural relaxation for minimization problems over the integral polyhedron $P(\F) := conv(\{\chi^F: \textit{$F \in \F$}\})$. The {\em blocking formulation} for $\F$ is:
\begin{equation}
\label{eqn:polyhedron}
P^*(\F) = \{z \in \R^V_{\geq 0}: z(B) \geq 1~~\forall B \in \B(\F^{min})=\B(\F)\}.
\end{equation}

\noindent
Clearly we have $P(\F) \subseteq P^*(\F)$.

\subsection{Reducing to blocking families}

Now consider an arbitrary set family $\F$ over $V$. We may define its {\em upwards closure}
by $\F^{\uparrow}=\{F':  F \subseteq F' \textit{ for some $F \in \F$}\}$. In this section we argue that
in order to solve a monotone optimization problem over sets in $\F$ it is often sufficient to work
over its upwards-closure.

As already  noted $\B(\F)=\B(\F^{\uparrow}) = \B (\F^{min})$ and hence one approach
 is via the blocking formulation $P^*(\F)=P^*(\F^{\uparrow})$. This requires two ingredients. First, we  need
 a separation algorithm for the blocking relaxation, but indeed this is often available for many
 natural families such as spanning trees, perfect matchings, $st$-paths, and vertex covers.   The second ingredient  needed is the ability to turn an integral solution $\chi^{F'}$ from  $P^*(\F^{\uparrow})$ or $P(\F^\uparrow)$
 into an integral solution $\chi^F \in P(\F)$. We now argue that this is the case if a polytime separation algorithm is available for the blocking relaxation $P^*(\F^{\uparrow})$ or for the polytope $P(\F):= conv(\{\chi^F: \textit{$F \in \F$}\})$.

 For a polyhedron $P$,  we denote its {\em dominant} by
$P^{\uparrow} := \{z:  z \geq x \textit{~for some~} x \in P \}$.
The following observation is straightforward.

\begin{claim}
	\label{claim:lattice-points}
	Let $H$ be the set of vertices of the hypercube in $\R^V$. Then
	$$H \cap P(\F^\uparrow) = H \cap P(\F)^\uparrow = H \cap P^*(\F^\uparrow).$$
	In particular we have that $\chi^S \in P(\F)^\uparrow \iff \chi^S \in P^*(\F^\uparrow)$.
\end{claim}

We can now use this observation to prove the following.

\begin{lemma}
	\label{lem:dominant-reduction}
	Assume we have a separation algorithm for $P^*(\F^\uparrow)$. Then for any $\chi^{S} \in P^*(\F^\uparrow)$ we can find in polytime $\chi^{M} \in P(\F)$ such that $\chi^{M} \leq \chi^{S}$.
\end{lemma}
\begin{proof}
	Let $S=\{1,2,\ldots,k\}$. We run the following routine until no more elements can be removed:
	\vspace*{10pt}
	
	For $i \in S$\\
	\hspace*{20pt} If $\chi^{S-i} \in P^*(\F^\uparrow)$ then $S=S-i$
	
	\vspace*{10pt}
	Let $\chi^M$ be the output. We show that $\chi^M \in P(\F)$. Since $\chi^M \in P^*(\F^\uparrow)$, by Claim \ref{claim:lattice-points} we know that $\chi^M \in P(\F)^\uparrow$. Then by definition of dominant there exists $x\in P(\F)$ such that $x\leq \chi^M \in P(\F)^\uparrow$. It follows that the vector $x$ can be written as $x = \sum_{i} \lambda_{i} \chi^{U_i}$ for some $U_i \in \F$ and $\lambda_i \in (0,1]$ with $\sum_i \lambda_i =1$. Clearly we must have that $U_i \subseteq M$ for all $i$, otherwise $x$ would have a non-zero component outside $M$.  In addition, if for some $i$ we have $U_i \subsetneq M$, then there must exist some $j \in M$ such that $U_i \subseteq M-j \subsetneq M$. Hence $M-j \in \F^{\uparrow}$, and thus $\chi^{M-j} \in P(\F)^\uparrow$ and $\chi^{M-j} \in P^*(\F^\uparrow)$. But then when component $j$ was considered in the algorithm above, we would have had $S$ such that $M \subseteq S$ and so $\chi^{S-j} \in P^*(\F^\uparrow)$ (that is $\chi^{S-j} \in P(\F)^\uparrow$), and so $j$ should have been removed from $S$, contradiction.
\qed \end{proof}

We point out that for many natural set families $\F$ we can work with the relaxation $P^*(\F^\uparrow)$ assuming that it admits a separation algorithm.
Then, if we have an algorithm which produces  $\chi^{F'} \in P^*(\F^\uparrow)$ satisfying some approximation guarantee for a monotone problem, we can use Lemma \ref{lem:dominant-reduction} to construct in polytime $F \in \F$ which obeys the same guarantee.

Moreover, notice that for Lemma \ref{lem:dominant-reduction} to work we do not need an actual separation oracle for $P^*(\F^\uparrow)$, but rather all we need is to be able to separate over $0-1$ vectors only. Hence, since the polyhedra $P^*(\F^\uparrow), \, P(\F^\uparrow)$ and $P(\F)^\uparrow$ have the same $0-1$ vectors (see Claim \ref{claim:lattice-points}), a separation oracle for either $P(\F^\uparrow)$ or $P(\F)^\uparrow$ would be enough for the routine of Lemma \ref{lem:dominant-reduction} to work. We now show that this is the case if we have a polytime separation oracle for $P(\F)$. The following result shows that if we can separate efficiently over $P(\F)$ then we can also separate efficiently over the dominant $P(\F)^\uparrow$.

\begin{claim}
	If we can separate over a polyhedron $P$ in polytime,  then we can also separate over its dominant $P^\uparrow$ in polytime.
\end{claim}
\begin{proof}
	Given a vector $y$, we can decide whether $y \in P^\uparrow$ by solving
	\begin{align*}
	x + s = y \\
	x \in P \\
	s \geq 0.
	\end{align*}
	Since can we easily separate over the first and third constraints, and a separation oracle for $P$ is given (i.e. we can also separate over the set of constraints imposed by the second line), it follows that we can separate over the above set of constraints in polytime.
\qed \end{proof}

Now we can apply the same mechanism from Lemma \ref{lem:dominant-reduction} to turn feasible sets from $\F^{\uparrow}$ into feasible sets in $\F$.

\begin{corollary}
	\label{cor:dominant-reduction2}
	Assume we have a separation algorithm for $P(\F)^\uparrow$. Then for any $\chi^{S} \in P(\F)^\uparrow$ we can find in polytime $\chi^{M} \in P(\F)$ such that $\chi^{M} \leq \chi^{S}$.
\end{corollary}

We conclude this section by making the remark that if we have an algorithm which produces  $\chi^{F'} \in P(\F^\uparrow)$ satisfying some approximation guarantee for a monotone problem, we can use Corollary \ref{cor:dominant-reduction2} to construct $F \in \F$ which obeys the same guarantee.

\section{Relaxations for constrained submodular minimization}
\label{sec:relaxations}

Submodular optimization
 techniques for minimization on a set family have involved two standard relaxations, one being linear  \cite{goel2009approximability} and one being convex
\cite{chekuri2011submodular,iwata2009submodular,iyer2014monotone}.
 We introduce the latter in this section.

\subsection{A convex relaxation}
\label{sec:LE-def}

We will be working with upwards-closed set families $\F$, and their blocking relaxations $P^*(\F)$.
 As we now work with arbitrary vectors $z \in [0,1]^n$,  we must specify how our objective function $f(S)$
behaves on all points $z \in P^*(\F)$. Formally, we call $g:[0,1]^V \rightarrow \R$ an {\em extension} of $f$ if $g(\chi^S) = f(S)$ for
each $S \subseteq V$.
For a submodular objective function $f(S)$ there can be many extensions of
$f$ to $[0,1]^V$ (or to $\R^V$).
The most popular one has been  the so-called {\em Lov\'asz Extension} (introduced
in \cite{lovasz1983submodular}) due to several of its desirable properties.

We present two of several equivalent definitions for the Lov\'asz Extension.
Let  $0   <  v_1 < v_2 < ... < v_m \leq 1$  be the distinct
positive values taken in some vector $z \in [0,1]^V$. We also define $v_0=0$ and $v_{m+1}=1$ (which may be equal to $v_m$).
Define for each $i \in \{0,1,...,m\}$ the set $S_i=\{ j:  z_j > v_i\}$. In particular, $S_0$ is the support of $z$
and  $S_m=\emptyset$.
One then defines:
\[
f^L(z) =  \sum_{i=0}^{m}   (v_{i+1}-v_i) f(S_i).
\]
It is not hard to check that the following is an equivalent definition of $f^L$ (e.g. see \cite{chekuri2011submodular}).
For a vector $z \in [0,1]^V$ and $\theta \in [0,1]$, let $z^\theta := \{v \in V: z(v)\geq \theta\}$.
We then have that
$$f^L(z) = \E_{\theta \in [0,1]} f(z^\theta) = \int_{0}^{1} f(z^\theta) d\theta,$$
where the expectation is taken over the uniform distribution in $[0,1]$.

\begin{lemma}
[Lov\'asz \cite{lovasz1983submodular}]
The function $f^L$ is convex if and only if $f$ is submodular.
\end{lemma}

One could now attack constrained submodular minimization by solving the problem  
\begin{equation}
\label{SA-LE-appendix}
(\mbox{SA-LE}) ~~~~
\min f^L(z): z \in P^*(\F),
\end{equation}
and then seek rounding methods for the resulting solution. This is the approach used in \cite{chekuri2011submodular,iwata2009submodular,iyer2014monotone}.
We refer to the above as the {\em single-agent Lov\'asz extension formulation}, abbreviated as (SA-LE).

\subsection{Tractability of the single-agent formulation (SA-LE)}
\label{sec:combined}

In this section we show that one may solve (SA-LE) approximately as long as a polytime separation algorithm for $P^*(\F)$ is available. This is useful in several settings and in particular for our methods which rely on the multi-agent Lov\'asz extension formulation (discussed in the next section).

{\bf Polytime Algorithms.} One may apply the Ellipsoid Method to obtain a polytime algorithm which approximately minimizes a convex function over a polyhedron $K$
as long as various technical conditions hold.  For instance, one could require that there are two ellipsoids $E(a,A) \subseteq  K \subseteq E(a,B)$ whose encoding
descriptions are polynomially bounded in the input size for $K$. We should also have polytime (or oracle) access to the convex objective function
defined over ${\bf R}^n$. In addition, one must be able to polytime solve  the subgradient problem for $f$.\footnote{For a given $y$, find a subgradient of $f$ at $y$.} One may check that the subgradient problem is efficiently solvable for Lov\'asz extensions of polynomially encodable submodular functions.
We call $f$ {\em polynomially encodable} if the values $f(S)$ have encoding size bounded by a polynomial in $n$ (we always assume this for our functions).
If these conditions hold, then methods from \cite{grotschel2012geometric} imply that for any $\epsilon >0$ we may find
 an approximately feasible solution for $K$ which is approximately optimal.  By approximate here we mean
 for instance that the objective value is within  $\epsilon$ of the real optimum.  This can be done
in time polynomially bounded in $n$ (size of input say) and $\log \frac{1}{\epsilon}$. Let us give a few details for our application.

Our convex problem's feasible space is $P^*(\F)$ and it is easy to verify
that our optimal solutions will lie in the $0-1$ hypercube $H$. So we may define the feasible space to be $H$
and the objective function to be $g(z)=f^L(z)$ if $z \in H \cap P^*(\F)$
and $=\infty$ otherwise. (Clearly $g$  is convex in ${\bf R}^n$ since it is a pointwise maximum of two convex functions;
alternatively, one may define the Lov\'asz Extension on ${\bf R}^n$ which is also fine.)  Note that $g$ can be evaluated in polytime
by the definition of $f^L$ as long as $f$ is polynomially encodable.
We  can now easily find an ellipsoid inside $H$ and one containing $H$ each of which has poly encoding size.
We may thus solve the convex problem to within $\pm \epsilon$-optimality in time bounded by a polynomial in $n$ and $\log \frac{1}{\epsilon}$.

\begin{corollary}
	\label{cor:ellipsoid}
	Consider a class of problems $\F, f$ for which $f$'s are submodular and polynomially-encodable  in $n=|V|$.
	If there is a polytime separation algorithm for the family of polyhedra $P^*(\F)$,
	then the convex program (SA-LE) can be solved to accuracy of $\pm \epsilon$ in time bounded
	by a polynomial in $n$ and $\log \frac{1}{\epsilon}$.
\end{corollary}

\subsection{The multi-agent formulation}
\label{sec:ma extensions}

The single-agent formulation (SA-LE) discussed above has a natural extension to the multi-agent setting. 
This was already introduced in \cite{chekuri2011submodular} for the case $\F=\{V\}$.
\begin{equation}
\label{MA-LE-appendix}
(\mbox{MA-LE}) ~~~~
\min \sum_{i=1}^k f^L_i(z_i): z_1 + z_2 + \cdots + z_k \in P^*(\F).
\end{equation}
We refer to the above as the {\em multi-agent Lov\'asz extension formulation}, abbreviated as (MA-LE).

We remark that we can solve (MA-LE) as long as
we have polytime separation of $P^*(\F)$. This follows the  approach from the previous section (see Corollary \ref{cor:ellipsoid}) except our convex program now has $k$ vectors of variables $z_1,z_2, \ldots ,z_k$ (one for each agent) such that $z=\sum_i z_i$. 
This problem has the form $\{ \min g(w) : w \in W \subseteq {\bf R}^{nk} \}$
where $g$ is convex and $W$ is the full-dimensional convex body $\{w=(z_1,...,z_k): \sum_i z_i \in P^*(\F)\}$. Clearly we have a polytime separation routine for $W$, and hence we may apply Ellipsoid as in the single-agent case.

\begin{corollary}
	\label{cor:sep-sol-LP}
	Assume there is a polytime separation oracle for $P^*(\F)$. Then we can solve the multi-agent formulation (MA-LE) in polytime.
\end{corollary}

\begin{corollary}
	\label{cor:SA-MA-LP}
	Assume we can solve the single-agent formulation (SA-LE) in polytime. Then we can also solve the multi-agent formulation (MA-LE) in polytime.
\end{corollary}
\begin{proof}
	If we can solve (SA-LE) in polytime then we can also separate over $P^*(\F)$ in polynomial time. Now the statement follows from Corollary~\ref{cor:sep-sol-LP}.
\qed \end{proof}


\section{Dealing with some special types of nonmonotone objectives}
\label{sec:nonmonotone}

We present the proof of Theorem \ref{thm:sym-MSCA} discussed in Section \ref{sec:SA-MA-formulations}.

\begin{proof}[Theorem \ref{thm:sym-MSCA}]
	Let $z = \sum_{i \in [k]} z_i$ be a feasible solution to (MA-LE) and such that $\sum_{i \in [k]} z_i \geq \chi^U$ for some $U \in \F$.
	Consider the below CE-Rounding procedure originally described in the work of \cite{chekuri2011submodular} for the case $\F=\{V\}$.
	
	
	\begin{algorithm}[ht]
		\caption{CE-Rounding}
		\label{alg:CE-rounding}
		$S \leftarrow \emptyset$ \quad /* set of assigned elements */\\
		$S_i \leftarrow \emptyset$ for all $i \in [k]$ /* set of elements assigned to agent $i$ */\\
		\While{$S \notin \F$}{
			Pick $i\in [k]$ uniformily at random\\
			Pick $\theta \in [0,1]$ uniformily at random \\
			$S(i,\theta):= \{v \in V: z_i(v)\geq \theta\}$ \\
			$S_i \leftarrow S_i \cup S(i,\theta) $\\
			$S \leftarrow S \cup S(i,\theta) $\\
		}
		/* Uncross $S_1,S_2,\ldots,S_k$ */\\
		$S'_i \leftarrow S_i$ for all $i \in [k]$\\
		\While{there exist $i \neq j$ such that $S'_i \cap S'_j \neq \emptyset$}{
			\eIf{$h(S'_i) + h(S'_j - S'_i) \leq h(S'_i) + h(S'_j)$}{
				$S'_j \leftarrow S'_j - S'_i$\\}{
				$S'_i \leftarrow S'_i - S'_j$\\}
		}
		Output $(S'_1,S'_2,\ldots,S'_k)$
	\end{algorithm}
	
	
	It is discussed in \cite{chekuri2011submodular} and not difficult to see
	that the first while loop assigns all the elements from $U$ in $O(k \log |U|)$ iterations with high probability.
	Since $\F$ is upwards closed and $U \in \F$, this implies that the first while loop terminates in
	$O(k \log |U|)$ iterations with high probability.
	Moreover, it is clear that the uncrossing step takes a polynomial number of iterations.
	
	Let $S_1,S_2,\ldots,S_k$ be the output after the first while loop and $S'_1,S'_2,\ldots,S'_k$ the final output of the rounding. At each iteration of the first while loop, the expected cost associated to the random set $S(i,\theta)$ is given by $$\E_{i,\theta}[f_i(S(i,\theta))] = \frac{1}{k} \sum_{i=1}^k \E_{\theta}[f_i(S(i,\theta))]= \frac{1}{k} \sum_{i=1}^k f^L_i (z_i).$$ Hence, given the subadditivity of the objectives (since the functions are submodular and nonnegative), the expected cost increase at each iteration of the first while loop is upper bounded by $\frac{1}{k} \sum_{i=1}^k f^L_i (z_i)$. Since the first while loop terminates w.h.p. in $O(k \log |U|)$ iterations, it now follows by linearity of expectation that the total expected cost of $\sum_i f_i(S_i)$ is at most $O(\log |U|) \sum_{i=1}^k f^L_i (z_i)$.
	
	Finally, we use a result (see Lemma 3.1) from \cite{chekuri2011submodular} that guarantees that if $h$ is symmetric submodular then the uncrossing step of the rounding satisfies $\sum_i h(S'_i) \leq \sum_i h(S_i)$. Moreover, by monotonicity of the $g_i$ it is also clear that $\sum_i g_i(S'_i) \leq \sum_i g_i(S_i)$. Thus, we have that $S'_1 \uplus  S'_2 \uplus  \cdots \uplus  S'_k = S \in \F$ and $\sum_i f_i(S'_i) \leq \sum_i f_i(S_i) \leq O(\log |U|) \sum_{i=1}^k f^L_i (z_i)$ on expectation. This concludes the argument.
	\qed
\end{proof}

\section{Invariance under the lifting reduction}
\label{sec:Appendix-Invariance}

\begin{corollary}
	\label{Corollary p-matroid intersection}If $(V,\F)$ is a p-matroid
	intersection, then so is $(E,\F')$.\end{corollary}
\begin{proof}
	Let $\F=\cap_{i=1}^{p}\I_{i}$ for some matroids $(V,\I_{i})$. Then
	we have that
	\begin{align*}
	\F'= & \{S\subseteq E:S_{1}\uplus \cdots\uplus S_{k}\in\F\}\\
	= & \{S\subseteq E:S_{1}\uplus \cdots\uplus S_{k}\in\cap_{i=1}^{p}\I_{i}\}\\
	= & \{S\subseteq E:S_{1}\uplus \cdots\uplus S_{k}\in\I_{i},\forall i\in[p]\}\\
	= & \bigcap_{i\in[p]}\{S\subseteq E:S_{1}\uplus \cdots\uplus S_{k}\in\I_{i}\}\\
	= & \bigcap_{i\in[p]}\I_{i}'.
	\end{align*}
	Moreover, from Corollary \ref{cor:matroid-invariant} we know that
	$(E,\I_{i}')$ is a matroid for each $i\in[p]$, and the result follows. \qed \end{proof}

We now discuss some invariant properties with respect to the families $\F_i$.

\begin{proposition}
	\label{Claim F_i matroids}If $(V,\F_{i})$ is a matroid for each
	$i\in[k]$, then $(E,\H)$ is also a matroid.\end{proposition}
\begin{proof}
	Let $\M_{i}:=(\{i\}\times V,\I_{i})$ for $i\in[k]$, where $\I_{i}:=\{\{i\}\times S:S\in\F_{i}\}$.
	Since $(V,\F_{i})$ is a matroid, we have that $\M_{i}$ is also a
	matroid. Moreover, by taking the matroid union of the $\M_{i}$'s
	we get $(E,\H)$. Hence, $(E,\H)$ is a matroid.
	\qed \end{proof}

\begin{proposition}
	\label{claim ring families}
	If $\F_{i}$ is a ring family over $V$ for each $i\in[k]$, then
	$\H$ is a ring family over $E$.\end{proposition}
\begin{proof}
	Let $S,T\in\H$ and notice that $S\cup T=\biguplus_{i\in[k]}(\{i\}\times(S_{i}\cup T_{i})$)
	and $S\cap T=\biguplus_{i\in[k]}(\{i\}\times(S_{i}\cap T_{i})$).
	Since $\F_{i}$ is a ring family for each $i\in[k]$, it follows that
	$S_{i}\cup T_{i}\in\F_{i}$ and $S_{i}\cap T_{i}\in\F_{i}$ for each
	$i\in[k]$. Hence $S\cup T\in\H$ and $S\cap T\in\H$, and thus
	$\H$ is a ring family over $E$.
\qed \end{proof}

We saw in Section \ref{sec:lifting-reduction} that if the original functions $f_i$ are all submodular, then the lifted function $f$ is also submodular. Recall that for any set $S \subseteq E$ in the lifted space, there are unique sets $S_i \subseteq V$ such that $S = \uplus_{i \in [k]} (\{i\} \times S_i)$. We think of $S_i$ as the set of items assigned to agent $i$. In a similar way, given any vector $\bar{z} \in [0,1]^E$, there are unique vectors $z_i \in [0,1]^V$ such that $\bar{z} = (z_1,z_2,\ldots,z_k)$, where we think of $z_i$ as the vector associated to agent $i$. Our following result establishes the relationship between the values $f^M(\bar{z})$ and $\sum_{i \in [k]} f_i^M(z_i)$.

\begin{lemma}
	\label{lem:max-multilinear}
	Let the functions $f_i$ and $f$ be as described in the lifting reduction on Section \ref{sec:lifting-reduction}.
	Then for any vector $\bar{z}=(z_1,z_2,\ldots,z_k) \in [0,1]^E$, where $z_i \in [0,1]^V$ is the vector associated with agent $i$, we have that $f^M(\bar{z})=f^M(z_1,z_2,\ldots,z_k)=\sum_{i \in [k]} f_i^M(z_i)$.
\end{lemma}
\begin{proof}
	We use the definition of the multilinear extension in terms of expectations (see Section \ref{sec:max-SA-MA-formulations}).
	Recall that for a vector $z \in [0,1]^V$, $R^z$ denotes a random set that contains element $v_i$ independently with probability $z_{v_i}$. 
	We use $\prob_{z}(S)$ to denote $\prob[R^z=S]$. We then have
	\begin{eqnarray*}
		f^M (\bar{z}) & = & \E [f(R^{\bar{z}})] = \sum_{S \subseteq E} f(S) \prob_{\bar{z}}(S) \\
		& = & \sum_{S_1 \subseteq V} \sum_{S_2 \subseteq V} \cdots \sum_{S_k \subseteq V} [\sum_{i=1}^k f_i (S_i)] \cdot \prob_{(z_1,z_2,\ldots,z_k)}(S_1,S_2,\ldots,S_k) \\
		& = & \sum_{i=1}^k \sum_{S_1 \subseteq V} \sum_{S_2 \subseteq V} \cdots \sum_{S_k \subseteq V} f_i (S_i) \cdot \prob_{(z_1,z_2,\ldots,z_k)}(S_1,S_2,\ldots,S_k)\\
		& = & \sum_{i=1}^k \sum_{S_i \subseteq V} f_i (S_i) \sum_{S_j \subseteq V, j\neq i} \prob_{(z_1,z_2,\ldots,z_k)}(S_1,S_2,\ldots,S_k)\\
		& = & \sum_{i=1}^k \sum_{S_i \subseteq V} f_i(S_i) \prob_{z_i}(S_i)
		= \sum_{i=1}^k \E[f_i(S_i^{z_i})]
		=  \sum_{i=1}^k f^M_i(z_i).
	\end{eqnarray*}
	\qed \end{proof}

\end{document}